\documentclass[journal]{IEEEtran}
\usepackage{amsmath, amssymb}
\usepackage{url}
\usepackage{cite}

\usepackage{algorithm}
\usepackage[noend]{algpseudocode}

\usepackage{amsthm, graphicx, color}
\usepackage{caption, subcaption}

\newtheorem{theorem}{Theorem}
\newtheorem{lemma}{Lemma}

\definecolor{erdem}{rgb}{0,0,0}

\definecolor{superlightgray}{rgb}{.9,0.9,0.9}

\title{\huge Asynchronous Local Construction of Bounded-Degree Network Topologies Using Only Neighborhood  Information}

\begin{document}
\author{\color{erdem} Erdem Koyuncu, \emph{Member}, IEEE, and Hamid Jafarkhani, \emph{Fellow}, IEEE\thanks{\color{erdem}This work was supported in part by the DARPA GRAPHS program Award N66001-14-1-4061, and in part by the NSF Awards CCF-1814717 and CCF-1815339.}\thanks{\color{erdem}This work was presented in part \cite{confversion} at the IEEE Wireless Communications and Networking Conference, Apr. 2017.} \thanks{\color{erdem}E. Koyuncu is with the Department of Electrical and Computer Engineering, University of Illinois at Chicago. Email: ekoyuncu@uic.edu. H. Jafarkhani is with the Center for Pervasive Communications and Computing, University of California, Irvine. Email: hamidj@uci.edu.}}
\maketitle
\vspace{-45pt}
\begin{abstract}
We consider ad-hoc networks consisting of $n$ wireless nodes that are located on the plane. Any two given nodes are called neighbors if they are located within a certain distance (communication range) from one another. A given node can be directly connected to any one of its neighbors and picks its connections according to a unique topology control algorithm that is available at every node. Given that each node knows only the indices (unique identification numbers) of its one- and two-hop neighbors, we identify an algorithm that preserves connectivity and can operate without the need of any synchronization among nodes. Moreover, the algorithm results in a sparse graph with at most $5n$ edges and a maximum node degree of $10$. Existing algorithms with the same promises further require neighbor distance and/or direction information at each node. We also evaluate the performance of our algorithm for random networks. In this case, our algorithm provides an asymptotically connected network with $n(1+o(1))$ edges with a degree less than or equal to $6$ for $1-o(1)$ fraction of the nodes. We also introduce another asynchronous connectivity-preserving algorithm that can provide an upper bound as well as a lower bound on node degrees.
\end{abstract}

\section{Introduction} 
\subsection{Topology Control and its Objectives}
\label{secia}
Topology control is a powerful technique that is commonly used in ad-hoc wireless networks to reduce interference, provide  energy-efficient transmission, enable low-complexity routing, and so on \cite{topcontrolinadhocref, santisurvey, santibook}. It refers to the intelligent choice of connections between nodes so that the resulting graph representation of the network (with nodes and direct node-to-node connections respectively modeled as vertices and edges of the graph) satisfies certain properties such as connectivity.

We study the problem of topology control over plane networks with the disk-connectivity model. Specifically, we consider networks consisting of $n$ nodes that are indexed (and uniquely identified) by the natural numbers $1,\ldots,n$  with locations $x_1,\ldots,x_n\in\mathbb{R}^2$. A given node may only be directly connected to any other neighboring node that is within a certain communication range $R>0$ in a bidirectional manner.  As an example, a network consisting of $10$ nodes with no connections together with the communication range of Node $7$ is shown in Fig. 1(a). Node $7$ can be directly connected to any one of the nodes in its neighbor set $\{1,3,5,6\}$.

\begin{figure}[h]\vspace{5pt}
\begin{subfigure}[b]{\linewidth}
\begin{center}\scalebox{1}{\includegraphics{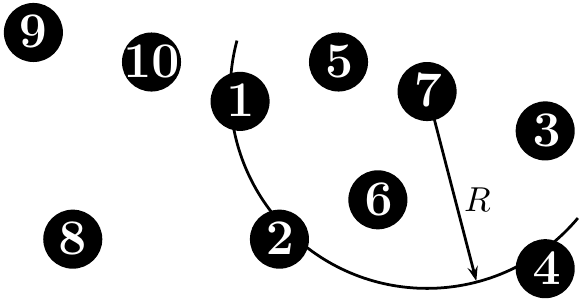}}\end{center}
\caption{The network with no connections. \\ \hspace{10pt} }
\end{subfigure}\vspace{5pt}
\begin{subfigure}[b]{\linewidth}
\begin{center}\scalebox{1}{\includegraphics{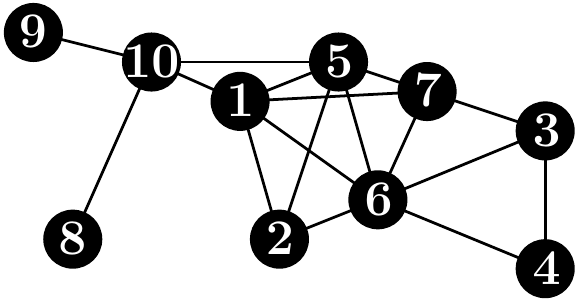}}\end{center}
\caption{The Gilbert graph corresponding to the node locations and the communication range illustrated in (a).}
\end{subfigure}\vspace{10pt}
\begin{subfigure}[b]{\linewidth}
\begin{center}\scalebox{1}{\includegraphics{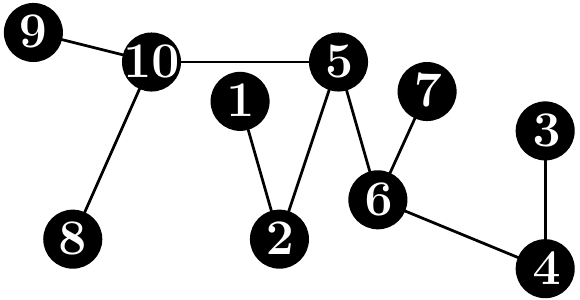}}\end{center}
\caption{A possible topology (a spanning subgraph of the Gilbert graph) generated by some topology control algorithm. }
\end{subfigure}
\caption{Instances of a network with $10$ nodes. Physical locations of the nodes are kept fixed throughout (a)-(c). The ``exact'' physical location of a given node is the center of the corresponding black disk.}
\label{networksfig}
\end{figure}

A special case is when all nodes within communication range are directly connected \cite{gilbert1}, which results in what we call the Gilbert graph $(\mathcal{V},g(\mathcal{V}))$ with
\begin{align}
g(\mathcal{W}) \triangleq \{(i,j):i,j\!\in\!\mathcal{W},\,i\!<\!j,\,|x_i-x_j| \!\leq\! R\},\,\mathcal{W}\!\subset\!\mathcal{V}.
\end{align}
Throughout the paper, $|\cdot|$ is the Euclidean metric. We note that Gilbert graphs are also often called unit-disk graphs whenever $R=1$ or with an appropriate normalization of node locations. As an example, Fig. 1(b) shows the Gilbert graph corresponding to the setup in Fig. 1(a). 

The primary goal of a topology control is then to provide a ``good''  spanning subgraph $(\mathcal{V},\mathcal{E})$ of the Gilbert graph $(\mathcal{V},g(\mathcal{V}))$. In this context, it is usually agreed upon that a good topology $(\mathcal{V},\mathcal{E})$ should satisfy the following properties:

\begin{enumerate}
\item \textbf{Connectivity:} The network $(\mathcal{V},\mathcal{E})$ is called connected if there is a path between any two distinct nodes in $\mathcal{V}$. It is clearly desirable to have a connected network so that information from one node may be conveyed to another (possibly through multiple hops) even if these two nodes are not directly connected. The algorithm/method that generates $(\mathcal{V},\mathcal{E})$ is called connectivity-preserving if $(\mathcal{V},\mathcal{E})$ is connected whenever $(\mathcal{V},g(\mathcal{V}))$ is connected. Conversely, since $(\mathcal{V},g(\mathcal{V}))$ is the largest feasible graph given the node locations, we note that $(\mathcal{V},\mathcal{E})$ can be connected only when  $(\mathcal{V},g(\mathcal{V}))$ is connected.
\item \textbf{Constant Stretch Factors:} Sometimes, the requirement of connectivity is further strengthened by imposing a low stretch factor as discussed in the following. Let us fix some $\alpha \geq 0$, and assign the weight $|x_i-x_j|^{\alpha}$ to every $(i,j)\in g(\mathcal{V})$. The cost of a given path is defined as the sum of the weights of the edges that appear in the path. For any $i,j\in\mathcal{V}$ with $i \neq j$, if the path with the lowest cost connecting Nodes $i$ and $j$ in $(\mathcal{V},\mathcal{E})$ is no more than $b$ times that in $(\mathcal{V},g(\mathcal{V}))$, we say that $(\mathcal{V},\mathcal{E})$ has an ``$\alpha$-stretch factor'' of $b$. The $0$- and $1$-stretch factors are commonly referred to as hop- and distance- stretch factors, respectively. For example, the network in Fig. 1(c) has a hop-stretch factor of $4$: Two nodes that are $h$-hop apart in Fig. 1(b) are no more than $4h$-hops apart in Fig. 1(c). Note that even if a given topology is connected, its (minimum) hop-stretch factor can be as large as $n-1$, while its $\alpha$-stretch factors for $\alpha > 0$ can be arbitrarily large. The stretch factors of a graph are related to the energy required for transmission of information from one node to another \cite{rodoplu, lmst}. It is thus desirable to have constant stretch factors that are as small as possible.

\item \textbf{Sparseness:} The network $(\mathcal{V},\mathcal{E})$ is called a sparse network if $|\mathcal{E}| \leq c n $ for some constant $c \geq 0$. A sparse network is desirable as the computational complexity of routing grows with the number of edges in the network. We note that a Gilbert graph is, in general, not a sparse network as it can have as many as $\frac{1}{2}n(n-1)$ edges. 
\item \textbf{Constant Maximum Degree:} The degree of a Node $i\in\mathcal{V}$ is the number of nodes that are directly connected to Node $i$. The existence of nodes with high degrees is not desirable in wireless networks due to several practical issues such as radio interference{\color{erdem}\cite{gupta1, vaze1, koyuncu1}}. In fact, in practice, a given wireless node can be connected to at most a finite number of nodes at any given time, merely due to the fact that there can be at most a finite number of non-interfering frequency bands. In some cases, physical limitations of wireless devices themselves necessitate degree restrictions. Also, several communication standards have ``built-in'' node degree constraints; for example, in Bluetooth networks, a master node can be connected to at most $7$ active slave nodes at a given time \cite{bluetoothsurvey}. It is thus desirable that the degree of every node in $(\mathcal{V},\mathcal{E})$ is no more than a constant $d \geq 0$ that is independent of $n$. The maximum degree of a general Gilbert graph can be as high as $n-1$. We also note that a graph with constant maximum degree is also a sparse graph, but the converse is not true in general.
\end{enumerate}

Consider now the problem of generating good topologies with the four desired properties listed above. The first two and the last two of the properties are mutually complementary. On the other hand, satisfying any one of the first two properties together with any one of the last two properties represent two contradicting goals. Also, in practice, it is unreasonable to expect the topology to be generated and imposed upon by a decision center that has global knowledge on the nodes' physical locations and identities. Instead, the topology should ideally be generated locally in a distributed fashion with every node picking its own connections using as little information from its neighboring nodes as possible. The design of such practical topology control methods has thus been a major avenue of research in the field of networking. We provide an overview of some of the relevant literature next.

\subsection{Related Work}
\label{secib}
There has been numerous works on local construction of topologies with some or all of the four desired properties as listed in Section \ref{secia}. Several well-known structured graphs have been a source of inspiration for many of these studies. For example, topology control algorithms have been inspired by Delaunay triangulations \cite{hu1, wang1, gao1journ, li4, bluedut1}, Gabriel graphs \cite{song1, li3, li2}, the minimum spanning tree \cite{yeniref1, lmst}, Yao graphs \cite{song1, li3, li2, watt2},  relative neighborhood graphs \cite{li3, yeniref1, li1}, or maximal independent sets \cite{yeniref2} and several new topologies have been discovered and analyzed in the process. There are also other approaches to topology control; for example, algorithms that allow the adjustment of the communication ranges of each node have been investigated \cite{ramanat1}. We refer to  \cite{topcontrolinadhocref, santisurvey, santibook} for a general detailed treatment of topology control including other algorithms.

One can conclude from the definitions in Section \ref{secia} that a constant stretch factor implies the preservation of connectivity, and a constant maximum node degree implies sparseness. Hence, the most difficult topologies to construct have been the ones with constant stretch factors and maximum node degrees. In fact, existing algorithms that can provide such topologies (see, e.g. \cite{wang1, song1}) require each node to know its exact geographical location (e.g. via GPS) as well as the locations of their neighbors and rely on several complex stages of message exchange between neighboring nodes. Other works have considered scenarios where each node has limited information about its neighbors. A notable algorithm is the XTC algorithm in \cite{watt1}, where each node is assumed to only know its distance to its neighbors as well as a certain ordering of its one- and two-hop neighbors. The XTC algorithm can provide a connected network with constant maximum degree. Another example is the CBTC algorithm in \cite{watt2} which can operate with neighbor direction information at the nodes.

The availability of node geographical information (in the form of direction, distance, or both) has been a common assumption in all the above works on topology control. The acquisition and communication of geographical information, especially exact geographical information, are however both non-trivial tasks in practice. It is thus desirable to drop the requirement of geographical information entirely and focus on algorithms that can operate only with neighborhood  information. {\color{erdem} Some of the effort in this context has focused on achieving sparse almost sure connectivity instead of preserving connectivity whenever possible; see e.g. the $k$-Neigh protocol of \cite{kneigh} that is based on \cite{xue1}, or the random Bluetooth networks analyzed in \cite{broutin1}. These works do not consider node degree restrictions. On the other hand, an XTC-like algorithm that does not rely on distance information has been proposed in \cite{islam1}, but it can only preserve connectivity without any guarantees on sparseness or node degrees. 

Another approach to position-unaware topology control is to utilize the connected dominating set (CDS) \cite{cds1, cds2, cds3, yeniref4, yeniref5, yeniref6} of the network. It has been shown in \cite{yeniref4} that by using only neighborhood  information, one can construct connected sparse topologies via a minimum or close-to-minimum CDS. It is not clear, however, how to obtain a degree-bounded topology using the idea of a CDS with neighborhood  information only. For example, \cite{bluedut1, yeniref4, yeniref5} require extra position information at each node to obtain a CDS-based topologies with bounded node degrees. 

Position-unaware topology control has also been a major focus of research on Bluetooth scatternet formation \cite{bluetoothsurvey} with several proposed algorithms such as BlueStars \cite{bluestars}, BlueMesh \cite{bluemesh}, BlueMIS \cite{bluemis}, and BSF-UED \cite{bsfued}. Some of these algorithms can provide degree-bounded topologies, but the degree bound holds for only the master nodes of the network and not for all the nodes of the network. In this context,
construction of network topologies with a constant degree bound at every node and without position information at nodes has been described \cite{bluemis} as ``an interesting and major open problem in the area.'' In fact, in this paper, we solve the very same problem. Next, we describe the properties of our solution.

\subsection{Our Contributions}
In this work, we assume that each node only knows the indices of its one- and two-hop neighbors without any extra geographical information. Under this restriction, we design a local algorithm that preserves connectivity, results in a sparse network with at most $5n$ edges, and meanwhile guarantees a constant maximum node degree of $10$. {\color{erdem} With the same restrictions, to the best of our knowledge, there is no existing local algorithm that can provide connectivity with bounded degree.} We also present an average case evaluation of our algorithm and show that the algorithm can in fact preserve connectivity with the almost-optimal amount of $n(1+o(1))$ edges and a degree less than or equal to $6$ for $1-o(1)$ fraction of the edges. We also note that the same algorithm can be applied to the scenario where the nodes are located on $\mathbb{R}^d,\,d\in\{1,2,\ldots\}$, and will similarly provide degree-bounded sparse connected topologies.

Our results show that neighborhood  information by itself can provide several fundamental properties that are desirable in ad-hoc wireless network topologies. One can however only achieve so much with only neighborhood  information. In fact, as we shall discuss in more details later, it is not difficult to show that  for any $\alpha  > 0$, no algorithm, even with a global knowledge of the network neighborhood  information, can provide a constant $\alpha$-stretch factor and a sparse network at the same time. Similarly, no algorithm can provide a constant hop-stretch factor and a constant degree bound at the same time. {\color{erdem} We shall emphasize that there are many algorithms in the existing literature that can guarantee connectivity with degree-bounded nodes and even finite stretch factors. Such algorithms were discussed in Section \ref{secib}. However, all of these algorithms require position information at the networking nodes. What distinguishes this work from the existing literature is that we present the first local algorithm that does not need any position information and can preserve connectivity with degree-bounded nodes.} We also introduce another asynchronous algorithm that provides both an upper bound and a lower bound on node degree.

We note that part of this work \cite{confversion} has been presented at the IEEE Wireless Communications and Networking Conference in April 2017. Compared to  \cite{confversion}, the current manuscript provides:
\begin{itemize}
\item the formal proofs for the average performance of the algorithm (the corresponding results in Theorem 2 of this paper were merely stated in \cite{confversion} without proof), 
\item a new section that discusses the achievability of stretch factors using only neighborhood information together with new corresponding simulation results,
\item comparisons with the existing algorithms in the literature such as XTC and $k$-Neigh, 
\item a new algorithm that provides a topology with guaranteed degree lower bounds, and finally, 
\item implementation details and communication complexity of our topology control algorithm.
 \end{itemize}

\subsection{Organization}
The rest of the paper is organized as follows: In Section \ref{secid}, we introduce the notation and conventions that will be used throughout the paper. In Section \ref{secii}, we present our topology control algorithm and formally prove its properties. In Section \ref{seciii}, we present an average case analysis of our algorithm. In Section \ref{secpracticalities}, we discuss the practicalities that are associated with our topology control algorithm. In Section \ref{seciv}, we investigate the achievability of constant stretch factors using only neighborhood information. In Section \ref{secdeglb}, we consider the construction of robust graphs that provide a degree lower bound in addition to a degree upper bound. In Section \ref{secv}, we present a numerical evaluation of our algorithms. Finally, in Section \ref{secvi}, we draw the main conclusions.

\subsection{Notation and Conventions}
\label{secid}
Given $i,j\in\mathcal{V}$ with $i \neq j$, we say that Nodes $i$ and $j$ are two \textbf{neighboring nodes}, or simply \textbf{neighbors} if $|x_i - x_j| \leq R$. Throughout the paper, we will only consider simple graphs (i.e., undirected graphs with no self-loops or multiple edges) of the form $(\mathcal{W},\mathcal{F})$ with $\mathcal{W}\subset\mathcal{V}$ and $\mathcal{F}\subset g(\mathcal{W})$. Given any such graph/network $(\mathcal{W},\mathcal{F})$, and any two indices/nodes $i,j\in\mathcal{W}$ with $i \neq j$, we say that Nodes $i$ and $j$ are \textbf{directly connected} in $(\mathcal{W},\mathcal{F})$ if $(i,j)\in\mathcal{F}$. Note that $(i,j)$ and $(j,i)$ will always represent the same edge. 

A \textbf{path} $\mathbf{p} \triangleq (p_1,\ldots,p_{|\mathbf{p}|})$ in the graph $(\mathcal{W},\mathcal{F})$ is a vector of distinct elements  of $\mathcal{W}$ such that $|\mathbf{p}| \geq 2$ and $(p_i,p_{i+1})\in\mathcal{F},\,\forall i\in\{1,\ldots,|\mathcal{P}|-1\}$. We say that Nodes $i$ and $j$ are \textbf{path-connected} in $(\mathcal{W},\mathcal{F})$ if there is a path $\mathbf{p}$ in $(\mathcal{W},\mathcal{F})$ with $p_1 = i$ and $p_{|\mathbf{p}|} = j$. A network $(\mathcal{W},\mathcal{F})$ is called \textbf{connected} if there is a path in $(\mathcal{W},\mathcal{F})$ between any two distinct nodes in $\mathcal{W}$. By definition, any two nodes that are directly connected are also path-connected.

A \textbf{connected component} of $(\mathcal{W},\mathcal{F})$ is a connected subgraph $(\mathcal{W}',\mathcal{F}')$ of $(\mathcal{W},\mathcal{F})$ such that (i) for any $i\in\mathcal{W}'$ and $j\in\mathcal{W}-\mathcal{W}'$, Nodes $i$ and $j$ are not path-connected in $(\mathcal{W},\mathcal{F})$, and (ii) for any $i,j\in\mathcal{W}'$, we have $(i,j)\in\mathcal{F} \!\implies\! (i,j)\in\mathcal{F}'$.

\section{The Main Algorithm}
\label{secii}
In this section, we present an algorithm that preserves connectivity and results in a sparse network with at most $5n$ edges and a maximum node degree of $10$. The setup in which the algorithm operates is as follows: Initially, we consider a network without any connections. The unique algorithm will be available to every node, and when ``run,'' will directly connect its ``host node'' (i.e., the node that is running the algorithm) to a certain subset of its host's neighboring nodes. All the  connections initiated by the algorithm will be bidirectional. Running the algorithm at every node exactly once will result in the topology with the aforementioned properties. Nodes will be able to run the algorithm in an arbitrary order, or simultaneously in a completely asynchronous fashion.

Let us now present the algorithm itself. A key definition we need is the notion of a \textbf{lesser neighborhood} of a node. For any $i\in\mathcal{V}$, we define the lesser neighborhood of Node $i$ as
\begin{align}
\mathcal{N}_i \triangleq \{j:j\in\mathcal{V},\,j<i,\,|x_i - x_j| \leq R\}.
\end{align}
Thus, the lesser neighborhood of Node $i$ are neighbors of Node $i$ whose indices are less than $i$. 

We recall from Section \ref{secia} that the Gilbert graph generated by the vertex set $\mathcal{W} \subset \mathcal{V}$ is given by the graph $(\mathcal{W},g(\mathcal{W}))$, where $g(\mathcal{W}) \triangleq \{(i,j):i,j\in\mathcal{W},\,i<j,\,|x_i-x_j| \leq R\}$. In other words, when all the nodes in $\mathcal{W}$ that are within communication range are directly connected, we obtain the Gilbert graph $(\mathcal{W},g(\mathcal{W}))$ generated by $\mathcal{W}$. Consider now the Gilbert graph $(\mathcal{N}_i,g(\mathcal{N}_i))$ generated by the lesser neighborhood of Node $i$. Let $J_i$ denote the number of connected components of $(\mathcal{N}_i,g(\mathcal{N}_i))$. Since each connected component of a Gilbert graph is necessarily also a Gilbert graph, we can list the connected components of $(\mathcal{N}_i,g(\mathcal{N}_i))$ as $(\mathcal{N}_{ij},g(\mathcal{N}_{ij})),\,j=1,\ldots,J_i$, where $\mathcal{N}_{ij},\,j=1,\ldots,J_i$ are mutually disjoint subsets of $\mathcal{N}_i$ with $\bigcup_{j=1}^{J_i} \mathcal{N}_{ij} = \mathcal{N}_i$. For any set $\mathcal{A}$, let $\max \mathcal{A}$ denote the maximum element of the set $\mathcal{A}$. Our algorithm (at Node $i$) is then as shown as Algorithm \ref{mainalgo}.

\begin{algorithm}%[t]
\caption{The Main Topology Control Algorithm (at Node $i$)}
\begin{algorithmic}[1]
\State Connect to all nodes in the set $\{\max \mathcal{N}_{ij}:1 \leq j\leq J_i\}$.
\end{algorithmic}
\label{mainalgo}
\end{algorithm}

We have previously mentioned that the nodes may run the algorithm in arbitrary order as long as each node runs the algorithm exactly once. In fact, it is easily observed that the order in which the nodes run the algorithm does not affect the final topology as long as each node runs the algorithm \emph{at least} once. All the different possibilities in this context will lead to the same final topology that we shall refer to as $(\mathcal{V},\mathcal{A})$.

\subsection{An Example Run}
We now demonstrate how the algorithm operates over the setup in Fig. 1(a). Suppose that initially there are no connections in the network. We illustrate how the algorithm (when it runs at Node $6$) determines the direct connections to be initiated by Node $6$. The lesser neighborhood of Node $6$ is given by $\mathcal{N}_6 = \{1,2,3,4,5\}$, as shown in Fig. 2(a). Note that Node $6$ itself and its ``greater'' neighbor Node $7$ are not members of $\mathcal{N}_6$. The next step for Node $6$ is to calculate the Gilbert graph $(\mathcal{N}_6,g(\mathcal{N}_6))$ induced by $\mathcal{N}_6$. This graph is as shown in Fig. 2(b) and has $J_6 = 2$ connected components $(\mathcal{N}_{61},g(\mathcal{N}_{61}))$ and $(\mathcal{N}_{62},g(\mathcal{N}_{62}))$ where $\mathcal{N}_{61} = \{1,2,5\}$ and $\mathcal{N}_{62} = \{3,4\}$. Finally, we have $\max \mathcal{N}_{61} = 5$ and $\max \mathcal{N}_{62} = 4$, so that Node $6$ will initiate a connection to Nodes $4$ and $5$. The corresponding two undirected edges that will be added to the initial graph will be $(4,6)$ and $(5,6)$. 

\begin{figure}[h]\begin{center}
\begin{subfigure}[b]{\linewidth}
\begin{center}\scalebox{1}{\includegraphics{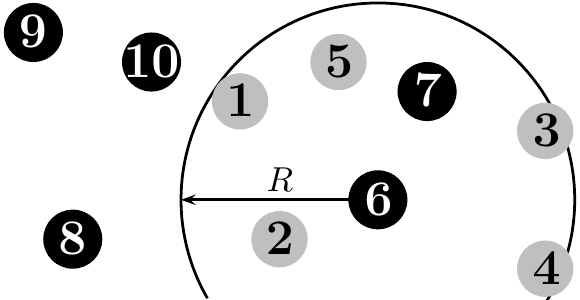}}\end{center}
\caption{Lesser neighbors $\mathcal{N}_6 = \{1,2,3,4,5\}$ of Node $6$. They are illustrated as gray disks.}
\end{subfigure}\hspace{10pt}
\begin{subfigure}[b]{\linewidth}
\begin{center}\scalebox{1}{\includegraphics{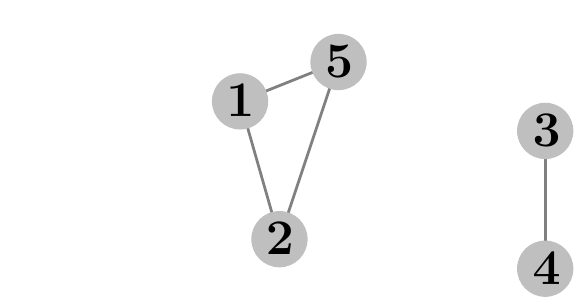}}\end{center}
\caption{The Gilbert graph induced by the lesser neighborhood of Node $6$.}
\end{subfigure}
\end{center}
\caption{The steps as to how Node $6$ determines its connections using the algorithm.}
\label{networksfig}
\end{figure}

In fact, running the algorithm at each node at least once results in the final network topology that we have previously shown in Fig. 1(c). For example, Node $1$, having no lower neighbors ($J_1 = 0$), will not initiate a connection to any other node. On the other hand, for Node $2$, we have $J_1 = 1$ with $\mathcal{N}_{21} = \{1\}$, so that Node $2$ will initiate a connection to Node $1$. Hence, Node $1$ in fact gets connected to Node $2$, even though it is not Node $1$ that initiates this connection.

\subsection{Analysis of the Algorithm}
\label{secanalysisofthealgorithm}
We now analyze the properties of the resulting topology $(\mathcal{V},\mathcal{A})$ generated by the algorithm. The following observation will be very useful for this purpose.
\begin{lemma}
In Algorithm \ref{mainalgo}, each node initiates at most $5$ connections. In other words, $J_i \leq 5$ for any $i\in\mathcal{V}$ (and for any given realization of node locations.).
\end{lemma}
\begin{proof}
Suppose that a given Node $i$ initiates connections to both Nodes $j_1$ and $j_2$; see Fig. 3 for an illustration. The obvious neighborhood conditions $|x_i - x_{j_1}| \leq R$ and $|x_i - x_{j_2}| \leq R$ hold. By the design of the algorithm, we also have $|x_{j_1} - x_{j_2}| > R$ (As otherwise, if $|x_{j_1} - x_{j_2}| \leq R$, Nodes $j_1$ and $j_2$ would belong to the same connected component, say $\mathcal{N}_{i\ell}$ for some $\ell\in\{1,\ldots,J_i\}$ of the Gilbert graph generated by $\mathcal{N}_i$. Then, since Node $i$ initiates a connection to only one of the nodes in $\mathcal{N}_{i\ell}$, it would then be absurd that it connects to both Nodes $j_1$ and $j_2$.). The three inequalities above imply that the edge $x_{j_1}x_{j_2}$ is the longest edge of the triangle $x_{j_1}x_ix_{j_2}$. This leads to the strict inequality $\theta_1 > 60^{\circ}$. Using the same arguments, we obtain $\theta_j > 60^{\circ},\,\forall j\in\{1,\ldots,J_i\}$. Now, assume the contrary to the statement of the lemma and suppose $J_i \geq 6$. We have $360^{\circ} = \theta_1 + \cdots + \theta_{J_i} > J_i 60^{\circ} \geq 360^{\circ}$. This is a contradiction that proves the lemma.
\end{proof}
\begin{figure}[h]
\center
\scalebox{1.1}{\includegraphics{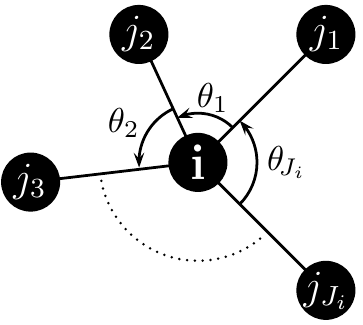}}
\caption{Figure for the proof of Lemma 1.}
\end{figure}

The following theorem is then the main result of this paper.
\begin{theorem}
\label{maintheorem}
The graph $(\mathcal{V}\!,\mathcal{A})$ is connected if and only if the Gilbert graph $(\mathcal{V}\!,g(\mathcal{V}))$ is connected. Moreover, we have $|\mathcal{A}| \leq 5n$ and the degree of each node in  $(\mathcal{V},\mathcal{A})$ is no more than $10$.
\end{theorem}
\begin{proof}
For the statement regarding connectivity, we only need to prove the ``if'' part with the ``only if'' part being trivial. Suppose $(\mathcal{V},g(\mathcal{V}))$ is connected. Then, for any given two nodes in $\mathcal{V}$, there is a path in $(\mathcal{V},g(\mathcal{V}))$ that connects these two nodes with each edge in the path consisting of two neighboring nodes. To show that $(\mathcal{V},\mathcal{A})$ is connected, it is thus sufficient to show that any two neighboring Nodes $i$ and $j$  are path-connected in $(\mathcal{V},\mathcal{A})$. To prove this, we may assume that $i < j$ without loss of generality. First, note that if $i=j-1$, then, by design, Node $j$ will initiate a connection to Node $i$ and Nodes $i$ and $j$ will be path-connected. Otherwise, $\exists k\in\mathcal{V}$ with $i < k < j$ such that (i) Node $j$ initiates a connection to Node $k$, and (ii) there is a path $\mathbf{p}$ in $(\mathcal{V},g(\mathcal{V}))$ connecting Node $k$ to Node $i$ such that the index of each node in $\mathbf{p}$ is no more than $k \leq j-1$. It is then sufficient to show that any two distinct neighbor nodes that appear in $\mathbf{p}$ are path-connected in $(\mathcal{V},\mathcal{A})$. On the other hand, to prove this latter claim, it is sufficient to show that any two neighboring Nodes $i'$, $j'$ with $i'<j' \leq j-1$ are path connected in $(\mathcal{V},\mathcal{A})$. 

By above arguments, we have established the following statement: Any two neighboring Nodes $i$ and $j$ with $i < j$ are path-connected in $(\mathcal{V},\mathcal{A})$ if either $i=j-1$ or any two neighboring Nodes $i'$ and $j'$ with $i' < j' \leq j - 1$ are path-connected in $(\mathcal{V},\mathcal{A})$. This statement describes a finite descent that immediately leads to the path-connectedness of Nodes $i$ and $j$. In fact, applying the statement on itself, any two neighboring Nodes $i$ and $j$ with $i < j$ are path-connected in $(\mathcal{V},\mathcal{A})$ if either $i=j-1$, or $i=j-2$, or any two neighboring Nodes $i'$ and $j'$ with $i' < j' \leq j - 2$ are path-connected in $(\mathcal{V},\mathcal{A})$. Hence, any two neighboring Nodes $i$ and $j$ with $i < j$ are path-connected in $(\mathcal{V},\mathcal{A})$ if $i = j-k$ for some natural number $k$, which is clearly true. This concludes the proof of the claim on connectivity.

We now prove the rest of the claims. The inequality $|\mathcal{A}| \leq 5n$  follows immediately as each node initiates at most $5$ connections by Lemma 1. We now prove the degree bound. Let $i\in\mathcal{V}$. By design, a node with a lower index ($<i$) cannot initiate a connection to Node $i$. On the other hand, Node $i$ itself initiates at most $5$ connections. To show a maximum node degree of $10$, it is thus sufficient to show that there are at most $5$ nodes with a higher index ($>i$) initiating a connection to Node $i$. Assume the contrary and suppose there are $6$ or more such nodes. Two of these nodes, say Nodes $j$ and $k$ (with $j<k$ without loss of generality) should then be within communication range as well as being within range of Node $i$. This implies $\{i,j\} \subset \mathcal{N}_{k\ell}$ for some $\ell\in\{1,\ldots,5\}$ with $i\notin \mathcal{N}_{k\ell'}$ and $j\notin \mathcal{N}_{k\ell'}$ for $\ell'\neq \ell$. Since $\max C_{k\ell} \geq \max\{i,j\} = j > i$, and $i \notin  \mathcal{N}_{k\ell'}$ for $\ell' \neq \ell$, we have, in fact, $\max \mathcal{N}_{k\ell} \neq i$ for every $\ell$. This contradicts the fact that Node $k$ initiates a connection to Node $i$ and thus proves the degree bound. 
\end{proof}

The degree bound of $10$ is tight in the sense that there are certain realizations of node locations for which the resulting topology $(\mathcal{V},\mathcal{A})$ has a node with degree $10$. A minimal example is with $11$ nodes, $x_6 = [0\,\,0]$, and $x_i = R[\cos\frac{i\pi}{5}\,\,\sin\frac{i\pi}{5}],\,i\in\{1,\ldots,11\}-\{6\}$. It does not seem to be as trivial, however, to find node locations that result in as much as $5n$ edges. In fact, as we show in the next section, the number of edges in most connected topologies that are generated by the algorithm is closer to $n$ than $5n$. We will also show that the maximum node degree in most networks generated by the algorithm is $6$.

\newcommand{\algorithmintuition}{Several variations of Algorithm \ref{mainalgo} can be envisioned. Some of these variations also provide useful insights on how and why the algorithm provides a degree-bounded topology and preserves connectivity at the same time. In this context, we discuss here the variant where Node $i$ connects to one arbitrary node in each of the sets $\mathcal{N}_{ij},\,j=1,\ldots,J_i$ instead of connecting to the nodes $\max N_{ij},\,j=1,\ldots,J_i$ with the maximum indices. Using the same arguments as in the proof of Theorem \ref{maintheorem}, it is straightforward to show that the variant algorithm preserves connectivity and provides a sparse graph with at most $5n$ edges. However, it does not provide a degree-bounded graph in general: Suppose all $n$ nodes are mutually within communication range. Running the variant algorithm, all nodes (except Node $1$) may decide to connect to Node $1$, resulting in a degree of $n-1$ at Node $1$ in the final network topology. 

The variant algorithm demonstrates that connecting to each disconnected component of the Gilbert graph induced by the lesser neighborhood of a node (as in our algorithm) provides sparsity and preserves connectivity. Such a connection strategy is, however, not enough to provide a degree-bounded topology. The connections should be done in an intelligent manner so as not to overwhelm a given node with too many connections. In Algorithm \ref{mainalgo}, this is done through connecting the node with the maximum index in a given component.}\algorithmintuition

As a final remark, we note that Theorem 1 can be applied and extended to networks in higher (or lower) dimensions, i.e. for networks in $\mathbb{R}^d$ for any $d \geq 1$ with the same disk-connectivity model. In fact, let $\mu_d$ denote the maximum number of points that can be packed in the unit ball in $\mathbb{R}^d$ such that any two given distinct points are more than one unit apart. We have $\mu_1 = 2$, $\mu_2 = 5$ (as shown in Lemma 1), and it is not difficult to show that $\mu_d$ is finite for any $d \geq 3$. The exact same algorithm generates a connectivity-preserving topology with at most $\mu_d n$ edges with a maximum node-degree of $2\mu_d$. In fact, similar results can be proved for connectivity models different than the disk model provided that the model admits a similar packing property. 

\section{Average Case Evaluation}
\label{seciii}
Algorithm \ref{mainalgo}, in the ``worst cases,'' results in a topology with $5n$ edges and a maximum node degree of $10$. However, numerical results suggest that for most realizations of node locations, the resulting topology is in fact much sparser and most nodes have a degree less than or equal to $6$. We present an analytical justification of this phenomenon using random graphs.

In this section, we let the node locations $x_1,\ldots,x_n$ be independent and uniformly distributed on $[0,1]^2$ (instead of being arbitrary fixed points in $\mathbb{R}^2$ as has been the case in previous sections). For any given fixed realization of node locations, we may simply run our algorithm to obtain one fixed topology corresponding to the given locations. The random nature of the node locations however means that the resulting topology $(\mathcal{V},\mathcal{A})$ will also be random. We are then interested in the properties of the now-random graph $(\mathcal{V},\mathcal{A})$ (We use the same notation for fixed and random graphs as the difference will be obvious from the context.).

For the random Gilbert graph $(\mathcal{V},g(\mathcal{V}))$, Penrose \cite{penrose} has shown the extremely precise result that if $R^2 = \frac{\log n + \alpha}{\pi n}$, then 
\begin{align}
\label{oqwiepqowiepqw}
\mathrm{Pr}((\mathcal{V},g(\mathcal{V}))\mbox{ is connected}) \rightarrow e^{-e^{-\alpha}} \mbox{ (as $n\rightarrow\infty$).} 
\end{align}
Here, $\mathrm{Pr}(\cdot)$ represents the probability of an event,  $\log(\cdot)$ is the natural logarithm, and $e$ is the base of the natural logarithm. In particular, $\mathrm{Pr}((\mathcal{V},g(\mathcal{V}))\mbox{ is connected}) \rightarrow 1$  if and only if $\alpha\rightarrow\infty$. We consider here random networks with communication radii just asymptotically above the connectivity threshold obtained by Penrose. Our main result is the following theorem.
\begin{theorem}
\label{randtheorem}
Suppose $R^2/(\frac{\log n}{n}) \rightarrow \infty$, and consider the random network $(\mathcal{V},\mathcal{A})$. Then, 
\begin{align}
\forall \epsilon > 0,\,\mathrm{Pr}(|\mathcal{A}| \geq (1+\epsilon)n) \rightarrow 0. 
\end{align}
Moreover, let $d_{\leq 6}$ denote the fraction of vertices in $(\mathcal{V},\mathcal{A})$ with degree no more than $6$. We have
\begin{align}
\forall \epsilon > 0,\,\mathrm{Pr}(d_{\leq 6} \geq 1-\epsilon) \rightarrow 1.
\end{align}
\end{theorem}
\begin{proof}
For any given Node $i$ with index $i>\beta n$, where $0 < \beta < 1$, let $\mathcal{E}_{ij}$ be the event that the $j$th Circular Sector $S_{i,j}$ of Node $i$ does not contain a node with index less than $\beta n$ (See Fig. 4 for the definition and the illustration of the circular sectors of a given node.). Ignoring the edge effects (which can be shown to not change the final results), we have
\begin{align}
\mathrm{Pr}(\mathcal{E}_{ij}) = \left(1-\frac{\pi R^2}{12}\right)^{\beta n},\,\forall j\in\{1,\ldots,12\},\,\forall i > \beta n.
\end{align}
Now, let $\mathcal{E}_i$ be the event that Node $i$ (with, again, $i > \beta n$) has a circular sector that does not contain a node with index less than $\beta n$.  By a union bound, we have 
\begin{align}
\mathrm{Pr}(\mathcal{E}_{i}) \leq 12\left(1-\frac{\pi R^2}{12}\right)^{\beta n},\,\forall i > \beta n.
\end{align}
Consider now the connections initiated by Node $i$ when it runs the algorithm. Given the complement of event $\mathcal{E}_i$, for any $j\in\{1,\ldots,12\}$, Sector $S_{i,j}$ of Node $i$ contains at least one node, say Node $s_{i,j}$, with $s_{i,j} < \beta n$. Now, note that any two given nodes in any sector are clearly neighbors. Moreover, for any given $j\in\{1,\ldots,12\}$, Node $s_{i,j}$ is necessarily a neighbor of the nodes of its neighboring sectors. In particular, Node $s_{i,1}$ is a neighbor of Nodes $s_{i,2}$ and $s_{i,12}$, Node $s_{i,2}$ is a neighbor of Nodes $s_{i,1}$ and $s_{i,3}$, and so on. These imply that the Gilbert graph induced by the lower neighborhood of Node $i$ has only one connected component, so that Node $i$ initiates only one direct connection provided that $i > \beta n$ and event $\mathcal{E}_i$ does not occur. On the other hand, since at most $5$ other nodes can initiate a direct connection to Node $i$ (this was proved as part of the proof of Theorem 1), the degree of Node $i$ will be at most $6$ in all of the final topologies where $\mathcal{E}_i$ does not hold.

\begin{figure}[h]
\center
\scalebox{1.1}{\includegraphics{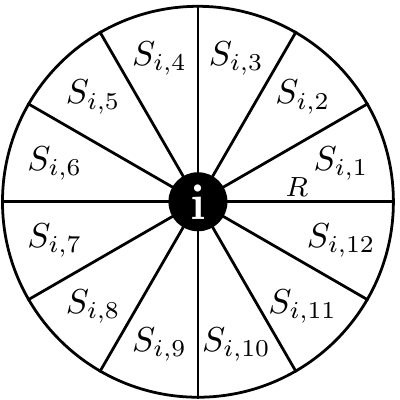}}
\caption{The circular sectors around the location $x_i$ of Node $i$ that are used in the proof of Theorem 2. Each sector includes its boundary. The central angle of each sector is equal to $30^{\circ}$. }
\end{figure}

If, further, event $\bigcup_{i=\beta n}^n \mathcal{E}_i$ does not occur, the network will have at most $5\beta n + (n-n\beta)$ edges (at most $5$ direct connections are initiated by nodes with indices less than $\beta n$, and only $1$ direct connection is initiated by nodes with indices greater than $\beta n$), with $(1-\beta)$-fraction of its nodes having degree no more than $6$. To prove the theorem, it is thus sufficient to show that $\mathrm{Pr}(\bigcup_{i=\beta n}^n \mathcal{E}_i) \rightarrow 0 $ with a suitable choice of $\beta$ such that $\beta \rightarrow 0$. In fact, using a union bound, and letting $\beta = \frac{24}{\pi} \frac{\log n}{nR^2}$, we have $\beta \rightarrow 0$, and 
\begin{multline}
\mathrm{Pr}\left(\bigcup_{i=\beta n}^n \mathcal{E}_i\right) \leq 12 n \left(1-\frac{\pi R^2}{12}\right)^{\beta n}  \\  \sim 12 n \exp\left(- \frac{\pi R^2 \beta n}{12} \right)   \sim \frac{12}{n} \rightarrow 0,
\end{multline}
where $\sim$ represents asymptotic equivalence. For the first equivalence, we have also assumed $R^2 \rightarrow 0$ without loss of generality (A network with a larger $R$ cannot have more edges or larger node degrees.). This concludes the proof.
\end{proof}
Hence, on average, Algorithm \ref{mainalgo} provides an extremely sparse connected network with $n(1+o(1))$ edges with a degree less than or equal to $6$ for $1-o(1)$ fraction of the nodes.

\section{Algorithm Implementation}
\label{secpracticalities}
In this section, we discuss the issues related to the implementation of Algorithm \ref{mainalgo}. We begin by describing a protocol for the implementation of the algorithm.
\subsection{A Protocol for Algorithm Implementation}
\label{secprot1}
\newcommand{\protone}{For the algorithm to work correctly, a given Node $i$ in the network only has to know its lesser neighbors $\mathcal{N}_i$ and the lesser neighbors $\mathcal{N}_j$ of each one of its lesser neighbors $j\in\mathcal{N}_i$. We weaken this statement by saying that each node only has to know its neighbors and the neighbors of its neighbors, i.e., its one- and two-hop neighbors. One way to implement the algorithm may then be via the following protocol that incorporates three rounds of inter-node communications: In the first round, each node may broadcast a ``Hello'' message (together with its index information) so that each node acquires the knowledge of its neighboring nodes. In the second round of communications, each node broadcasts the indices of its neighbors so that each node can also acquire the indices of each one of its neighboring nodes. Each node may then run the algorithm to determine the set of nodes to connect to; this step does not require any inter-node communication. In the final and third round of communications, each node broadcasts the indices of the nodes it has decided to connect to. Once the corresponding connections are established, the final topology is complete.}\protone
\subsection{Communication Complexity}
\label{secprot2}
\newcommand{\prottwo}{
The communication overhead of the above protocol can also be analyzed in an average sense. Suppose the $n$ nodes are distributed uniformly at random on $[0,1]^2$, as in Section \ref{seciii}. The index of each node can be represented via a binary word of length $O(\log n)$ bits. In the first round of communications, the message of each node is thus $O(\log n)$ bits, for a total of $O(n\log n)$ bits over the entire network. Given that the nodes are distributed on $[0,1]^2$, each node has $O(nR^2)$ neighbors on average, resulting in a per-node message length of $O(n \log n R^2)$ bits on average for the second round of communications. Finally, as each node initiates at most $5$ connections by Lemma 1, the per-node message length is $O(\log n)$ bits for the third round of communications. Each node thus sends a total of $O(\log n (1+nR^2))$ bits in total during the topology formation phase. In particular, setting $R^2 = \frac{\log n + \alpha}{\pi n}$ for some $\alpha \rightarrow \infty$ and $\alpha \in o(\log n)$, the network is asymptotically almost surely connected according to (\ref{oqwiepqowiepqw}), and the per-node message length to establish the topology is $O(\log^2 n)$ bits. Hence, the average per-node communication complexity for the establishment of the network topology is only polylogarithmic in the number of nodes.}\prottwo
\subsection{Node Identification}
\label{secnodeid}
We also note that in practice, a node may not carry any ``index information,'' at least not necessarily in the form of a natural number ranging from $1$ to $n$. Instead, each node may have a unique identification number (or a unique address) that can be used for indexing purposes. These identification numbers can be ordered, for example, lexicographically. Instead of the natural numbers with their standard order, the same algorithm can then operate over the node identification numbers with their lexicographical order. Hence, the (likely) possibility of ``unnatural'' node indices does not affect the way the algorithm operates or the final results. \newcommand{\nodeidcomment}{In general, we assume that each node is assigned its unique identification number during manufacturing, in a manner similar to the assignment of media access control (MAC) addresses. Hence, a separate  algorithm for node identification number assignment is not necessary.}\nodeidcomment

\section{The Unachievability of Constant Stretch Factors Using Neighborhood  Information}
\label{seciv}
We have shown the existence of a local topology control algorithm that can preserve connectivity with constant bounded maximum node degree using only one- and two-hop neighborhood  information. All the previous algorithms with the same promises in addition require geographical information (in the form of neighbors' distance/direction). As we have mentioned in Section \ref{secib}, some of these algorithms also guarantee constant $\alpha$-stretch factors, which provide a stronger notion of connectivity. Unfortunately, in the case of our algorithm, for any $\alpha \geq 0$, one can construct a specific realization of node locations such that the $\alpha$-stretch factor of the resulting topology can be made arbitrarily large. However, the stretch factors may be low with high probability, as we will show numerically in the next section. This begs the question of whether or not there exists another (better) algorithm that similarly uses only neighborhood  information and can provide constant stretch factors with bounded node degrees. In this section, we answer this question in the negative: There are no such algorithms even if one assumes global knowledge of neighborhood  information.

Let us first define the $\alpha$-stretch factors in a formal manner. Let us fix some $\alpha \geq 0$, and assign the weight $|x_i-x_j|^{\alpha}$ to every $(i,j)\in g(\mathcal{V})$. We can then model the cost of communicating from Node $i$ to Node $j$ on a given  spanning subgraph $(\mathcal{V},\mathcal{E})$ of $(\mathcal{V},g(\mathcal{V}))$ via the quantity 
\begin{align}
c_{\alpha}(i,j;\mathcal{E}) \triangleq \min_{\mathbf{p}} \sum_{i=1}^{|\mathbf{p}|-1} |x_{p_i} - x_{p_{i+1}}|^{\alpha},
\end{align}
where the minimization is over all paths in $(\mathcal{V},\mathcal{E})$ connecting Nodes $i$ and $j$. We let $c_{\alpha}(i,j;\mathcal{V},\mathcal{E}) = \infty$ if Nodes $i$ and $j$ are not path-connected in $(\mathcal{V},\mathcal{E})$. We say that the topology $(\mathcal{V},\mathcal{E})$ has an $\alpha$-stretch factor of $t$ if $c_{\alpha}(i,j;\mathcal{E}) \leq t c_{\alpha}(i,j;g(\mathcal{V}))$ for every $i,j\in\mathcal{V}$ with $i\neq j$. 

Let $\mathfrak{E}$ denote the collection of all possible sets of edges given the vertex set $\mathcal{V}$. A centralized (and deterministic) control of topology can then be modeled as a mapping $f:\mathfrak{E}\rightarrow\mathfrak{E}$ with the property that for every $\mathcal{E}\in\mathfrak{E}$, we have $ f(\mathcal{E}) \subset \mathcal{E}$. Operationally, we assume that a control center can somehow have access to (only) the entire neighborhood information $g(\mathcal{V})$, and declares $(\mathcal{V},f(g(\mathcal{V}))$ as the final topology. We have the following result regarding the hop-stretch factors.
\begin{theorem}
Suppose that $f$ preserves connectivity and results in constant maximum node degree of $d > 0$. Then, $f$ cannot provide a constant hop-stretch factor.
\end{theorem}
\begin{proof}
Suppose all the nodes are located within each other's communication range. Then, $g(\mathcal{V})$ is the complete graph $K_n$ of $n$ nodes where all the nodes are within $1$ hop of each other. Consider now the topology $f(K_n)$ generated by $f$. There are $n-1$ nodes that are at least $1$ hop away from Node $1$. Moreover, due to the degree bound provided by $f$, there are at least $n-d-1$ nodes that are at least $2$ hops away from Node $1$, and in general, at least $n-(1+d+\cdots+d^{k-1})$ nodes that are at least $k$ hops away from Node $1$. Hence, for any $k$ and $d$, if $n > 1+d+\cdots+d^{k-1}$, there exists a node with index $i\in\mathcal{V}$ such that Node $i$ is $k$ hops away from Node $1$ in $(\mathcal{V},f(K_n))$. On the other hand, Nodes $i$ and $1$ were only $1$ hop apart in the graph $(\mathcal{V},K_n)$. The hop-stretch factor with $f$ is thus at least $k$. The result now follows immediately as for any $d$, $k$ can be made arbitrarily large by considering a sufficiently large $n$. 
\end{proof}
A stronger version can be proved in the case of the $\alpha$-stretch factors for $\alpha > 0$.
\begin{theorem}
Suppose that $f$ preserves connectivity but with one or more edges missing from the Gilbert graph. Then, for any $\alpha > 0$, the mapping $f$ cannot provide a bounded $\alpha$-stretch factor.
\end{theorem}
\begin{proof}
Suppose all the nodes are located on a disk of radius $\epsilon R$ centered at the origin for some $\epsilon > 0$. We have $g(\mathcal{V}) = K_n$. Let $(i,j)$ denote one of the missing edges in $(\mathcal{V},f(K_n))$. Now, if $(i,j)$ is not path-connected in $(\mathcal{V},f(K_n))$, we have $c_{\alpha}(i,j;f(K_n)) = \infty$. The claim of the theorem then follows immediately as obviously, $c_{\alpha}(i,j;g(\mathcal{V})) \leq (2\epsilon)^{\alpha}$ and thus $c_{\alpha}(i,j;g(\mathcal{V}))$ is finite. Otherwise, let $P$ denote the set of all paths in $(\mathcal{V},f(K_n))$ that connect Node $i$ to Node $j$. Choose a node $k\notin\{i,j\}$ in the path with the minimum cost in $P$. Consider a change of locations where we move Node $k$ to $[R\,\,0]$ while keeping all the other nodes' locations fixed. We again have $g(\mathcal{V}) = K_n$ in this case, and the resulting topology is thus the same topology $(\mathcal{V},f(K_n))$ when all the nodes were located at the origin. However, we now have (for the new node locations) $c_{\alpha}(i,j;f(K_n)) > |x_i-x_k|^{\alpha}+|x_k-x_j|^{\alpha} = 2R^{\alpha}(1-\epsilon)^{\alpha}$ while $c_{\alpha}(i,j;g(\mathcal{V})) \leq (2\epsilon)^{\alpha}$. The $\alpha$-stretch factor provided by $f$ is then at least $\frac{2R^{\alpha}(1-\epsilon)^{\alpha}}{(2\epsilon)^{\alpha}}$, which can be made arbitrarily large by choosing a sufficiently small $\epsilon$. This concludes the proof.
\end{proof}

\section{Constructing Graphs with Degree Lower Bounds}\label{secdeglb}
In practice, it is also important to construct a robust graph, by e.g. providing a lower bound on the node degrees as well as an upper bound. This way, if a certain subset of communication links is broken, one can potentially use the remaining links to reach one node from another. 

Suppose that the degree of each node in a connected Gilbert graph $(\mathcal{V},g(\mathcal{V})$ is at least $\delta \geq 1$. In the following, we provide an asynchronous algorithm that preserves connectivity, and provides a lower bound of $\delta$ and a constant upper bound on the degree of every node. In this context, one naive algorithm that comes to mind is for each node to randomly add edges to guarantee a degree lower bound of $\delta$ after running Algorithm \ref{mainalgo}. Unfortunately, this approach does not lead to a constant upper bound on the degree of the networking nodes.

Let  $\overline{\mathcal{N}}_i \triangleq \{j:j\in\mathcal{V},\,j>i,\,|x_i - x_j| \leq R\}$ denote the greater neighborhood of Node $i$. We then consider the algorithm whose steps are provided in Algorithm \ref{newalgodeglb}. First, as shown by Line 1 of the algorithm, Node $i$ chooses the same nodes to connect to as in Algorithm \ref{mainalgo}. Next, in Lines 2 and 3, the remaining lesser neighbors are added to the list $\mathcal{C}_i$ of the nodes that Node $i$ will connect to, until the eventual degree of the node, $|\mathcal{C}_i|$, is guaranteed to be at least $\delta$. The priority is given to the lesser neighbors with the highest index.  Often, the set of lesser neighbors are not enough to satisfy the degree lower bound, in which case we add the greater neighbors of Node $i$ via Lines 4 and 5. This time, the neighbors with the smallest index are given priority.

\begin{algorithm}%[t]
\caption{Algorithm for Minimum Degree Guarantee (at Node $i$)}
\begin{algorithmic}[1]
\State $\mathcal{C}_i \leftarrow \{\max \mathcal{N}_{ij},\,1\leq j \leq J_i\}$. 
\While{$|\mathcal{C}_i| < \delta$ and $\mathcal{C}_i \neq \mathcal{N}_i$}
\State $\mathcal{C}_i \leftarrow \mathcal{C}_i \cup \{\max (\mathcal{N}_{i}-\mathcal{C}_i)\}$. 
\EndWhile
\While{$|\mathcal{C}_i| < \delta$ and $\mathcal{C}_i \neq \mathcal{N}_i \cup \overline{\mathcal{N}}_i$ }
\State $\mathcal{C}_i \leftarrow \mathcal{C}_i \cup  \{\min (\overline{\mathcal{N}}_{i}-\mathcal{C}_i)\}$. 
\EndWhile
\State Connect to all nodes in the set $\mathcal{C}_i$.
\end{algorithmic}
\label{newalgodeglb}
\end{algorithm}

In particular, for $\delta = 1$, Algorithm \ref{newalgodeglb} is identical to Algorithm 1. The following theorem provides the properties of the topologies generated by Algorithm \ref{newalgodeglb} for a general $\delta \geq 1$.
\begin{theorem}\label{deglbtheorem}
Suppose that the $n$-node graph $(\mathcal{V},g(\mathcal{V}))$ is connected with a degree of $d_i$ at Node $i$, where $i\in\{1,\ldots,n\}$. Then, the topology generated by Algorithm \ref{newalgodeglb} is connected with at most $\max\{\delta,5\}n$ edges. Moreover, for any given $i\in\{1,\ldots,n\}$, the degree of Node $i$ is at least $\min\{d_i,\delta\}$ but no more than $\max\{\delta,5\}+10\delta$.
\end{theorem}
\begin{proof}
Algorithm \ref{newalgodeglb} contains all edges that are provided by Algorithm \ref{mainalgo} due to Line 1. Therefore, since Algorithm \ref{mainalgo} provides a connected topology, so does Algorithm \ref{newalgodeglb}. The proof of Theorem 1 shows that Line 1 results in a set $\mathcal{C}_i$ with cardinality at most $5$, while Lines 3 and 5 add new nodes to $\mathcal{C}_i$ (one at a time) until $|\mathcal{C}_i| \geq \min\{d_i,\delta\}$. Therefore, $|\mathcal{C}_i| \leq \max\{\delta,5\}$. This implies that there are at most $\max\{\delta,5\}n$ edges in the final topology. We now provide the degree bounds. The fact that the final topology provides a degree lower bound of $\min\{d_i,\delta\}$ at every node is obvious by the design of the algorithm. We thus prove the degree upper bound. First, note that a given node initiates at most $\max\{\delta,5\}$ connections. We now determine the maximum number of nodes that establish a connection to Node $i$. We argue that there can be at most $5\delta$ greater neighbors of Node $i$ that connect to Node $i$. Assume the contrary. One can then find $\delta+1$ greater neighbors of Node $i$ with indices $j_1,\ldots,j_{\delta+1}$ that are all mutually within communication range and connect to Node $i$. Without loss of generality, suppose $i < j_1 < \cdots <j_{\delta+1}$. Since Node $j_{\delta+1}$ connects to Node $i$, it also necessarily connects to Nodes $j_1,\ldots,j_{\delta}$. There are now two possibilities when the algorithm is run on Node $j_{\delta+1}$. The first possibility is that the connection to Node $j_{\delta}$ is made via Line 1, and the remaining $\delta$ connections to Nodes $i,j_1,\ldots,j_{\delta-1}$ are made via Line 3 of the algorithm. This leads to a contradiction as Line 3 cannot increase the node degree beyond $\delta$. The second possibility is that all $\delta+1$ connections to Nodes $i,j_1,\ldots,j_\delta$ are made via Line 3 of the algorithm, which leads to a similar contradiction. A similar argument shows that there can be at most $5\delta$ lesser neighbors of Node $i$ that connect to Node $i$. As a result, the degree of each node in the final topology is no more than $\max\{\delta,5\}+10\delta$.
\end{proof}
In particular, if the degree of every node in the Gilbert graph is at least $\delta$, then Algorithm \ref{newalgodeglb} provides a degree lower bound of $\delta$ at every node, as desired. An interesting direction for further research is to find a better algorithm that improves the degree upper bound in Theorem \ref{deglbtheorem}, as for example, for $\delta=1$, Theorem 1 provides a better degree upper bound.

\section{Numerical Results}
\label{secv}
In this section, we present numerical simulations that confirm our analytical results and provide additional insights. We have run  Algorithm \ref{mainalgo} on a network with $n=1000$ nodes and initially no connections. Nodes are located independently and uniformly on $[0,1]^2$. We have considered the choices $\pi R^2 = \frac{N}{n}$ for $N\in\{10,20,30\}$. The parameter $N$ can be thought as a quantification of ``node density'' as any given node of the network then has roughly (ignoring the edge effects) $N$ neighbors on average. Regarding our specific choices for $N$, we note that the probabilities of connectivity for the associated Gilbert graphs are approximately $0.5654$, $0.9922$, and $0.9996$ for the choices $N=10$, $N=20$, and $N=30$, respectively. These values were obtained numerically by averaging over at least $5000$ node location realizations. Hence, the three different choices for $N$ represent the three different scenarios of ``mostly-disconnected,'' ``usually connected,'' and ``almost-always connected'' networks.

In Fig. \ref{edgesimfig}, we show the cumulative distribution functions (CDFs) of the normalized number of edges $\frac{|\mathcal{A}|}{n}$ for different values of $N$. In all the node location realizations we have simulated, the number of edges of the network never exceeded $1.14n$ or went below $1.02n$ for any choice of $N$. These observations are in agreement with the inequality $|\mathcal{A}|\leq 5n$ as proved by Theorem 1. In addition, the fact that the number of edges are very close to $n$ for any $N$ is in agreement with Theorem $2$, where we proved that the algorithm usually generates topologies with $n(1+o(1))$ edges, especially when $N$ is large. In fact, if $N$ were equal to infinity, all the nodes would be within the range of each other, and the algorithm would generate the line topology with only the $n$ edges $(1,2),(2,3),\ldots,(n-1,n)$. 

\begin{figure}[h]
\center
\scalebox{0.55}{\includegraphics{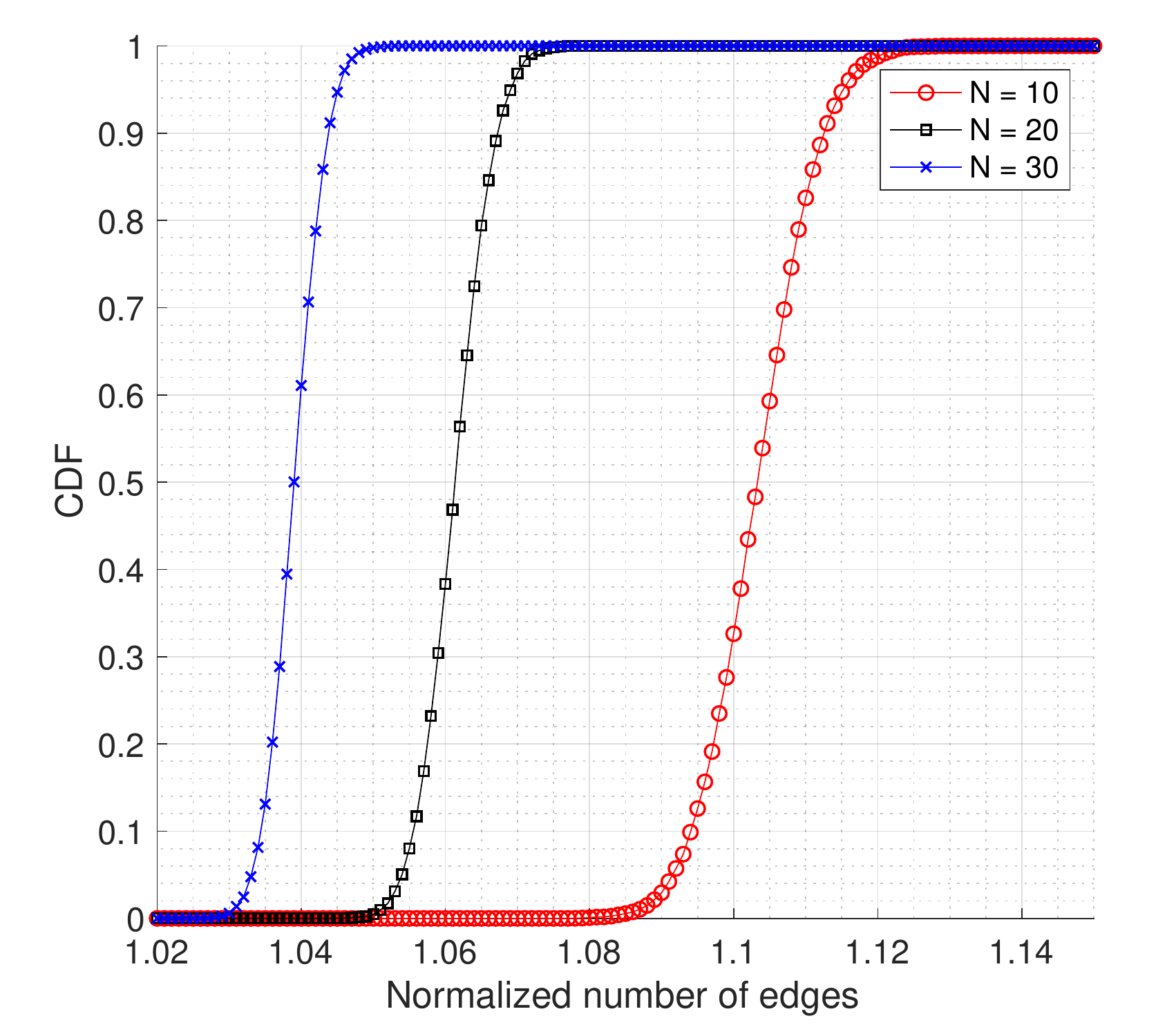}}
\caption{The CDFs of the normalized number of edges $\frac{|\mathcal{A}|}{n}$ for different node densities with Algorithm \ref{mainalgo}.}
\label{edgesimfig}
\end{figure}

In Fig. \ref{edgesimfigcomp}, we compare the number of edges provided by Algorithm 1 with other algorithms or topologies for a  node density of $20$. We consider the topologies generated by the XTC algorithm \cite{watt1}, Algorithm 2 with $\delta = 2$, the $k$-Neigh algorithm \cite{kneigh} for $k=6$, and the Gilbert graph. XTC is a powerful algorithm that can provide a degree-bounded topology with constant stretch factors. However, it requires neighbor distance information at the nodes. We recall from Section \ref{secdeglb} that Algorithm 2 for $\delta = 2$ guarantees (if at all feasible) a degree lower bound of $\delta = 2$ at every node, and a degree upper bound of $25$. In particular, for $\delta = 1$, Algorithm 2 is equivalent to Algorithm 1. According to the $k$-Neigh algorithm, each node connects to $k$ of its closest neighbors. The $k$-Neigh algorithm provides sparse topologies with typically low-degree nodes. However, it also requires neighbor distance information, and is not guaranteed to preserve connectivity. We have considered the choice $k=6$ for fairness in terms of the probability of connectivity: For $k=6$, the $k$-Neigh algorithm provides connectivity with probability $0.9904$, which is a negligible loss compared to the probability $0.9922$ of connectivity of the Gilbert graph. For $k=5$, the probability of connectivity with $k$-Neigh drops to $0.9681$. 

We can observe from Fig. \ref{edgesimfigcomp} that with probability $0.99$, Algorithm 1 provides the sparsest topology with (at most) $1.07n$ edges, followed by the XTC algorithm with $1.27n$ edges, Algorithm 2 for $\delta = 2$ with $1.88n$ edges, $k$-Neigh for $k=6$  with $3.59n$ edges, and finally the Gilbert graph with $9.64n$ edges. Thus, Algorithm 1 provides around $9$-fold reduction for the required number of edges for connectivity compared to the Gilbert topology. Also, compared to its closest competitor XTC, our algorithm reduces the required number of edges for connectivity by around $100\times\frac{1.27n-1.07n}{1.27n} \approx 16$ per cent. Moreover, unlike XTC, the reduction comes without the need for the extra neighbor distance information at networking nodes.

\begin{figure}[h]
\center
\scalebox{0.55}{\includegraphics{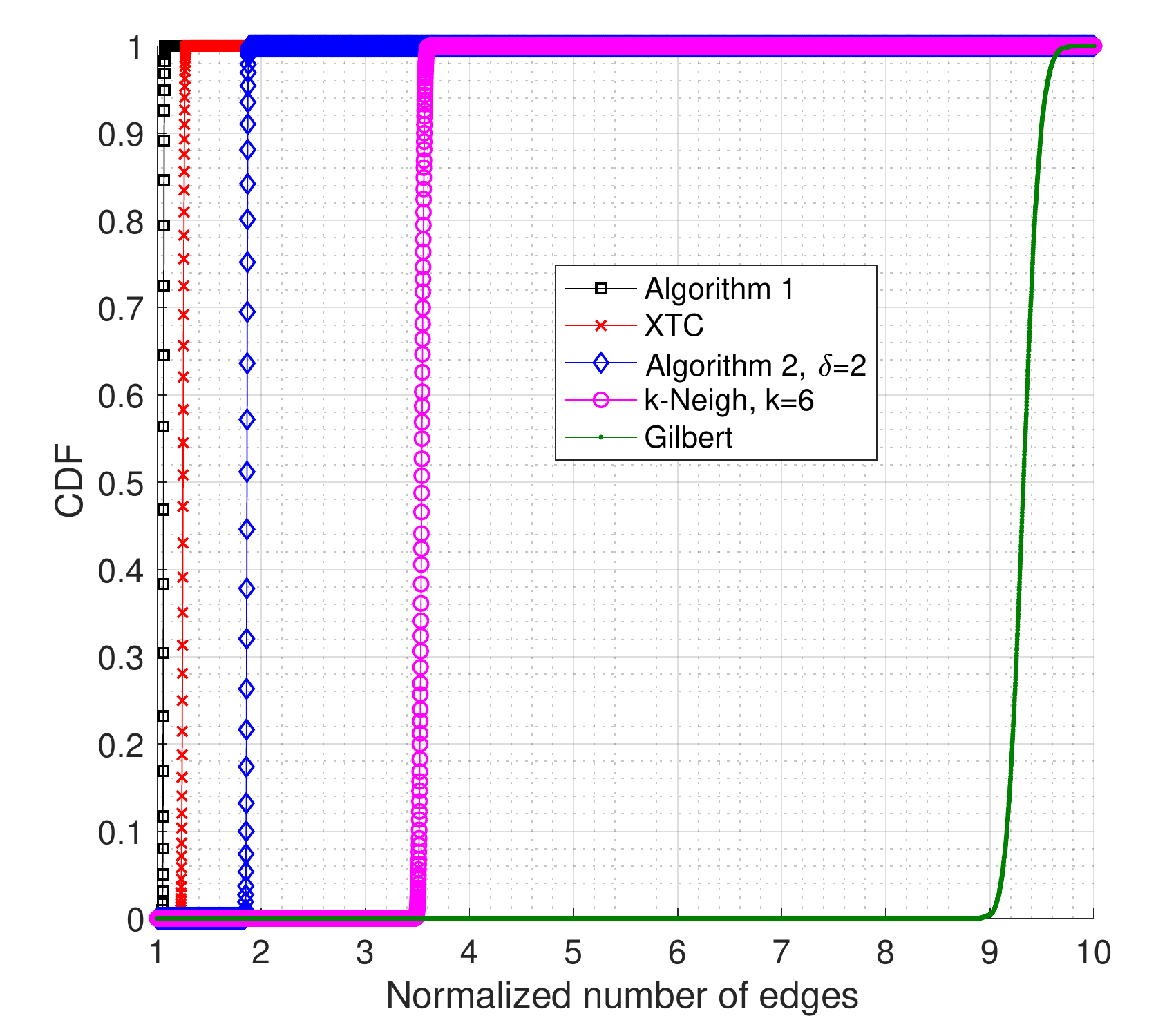}}
\caption{The CDFs of the normalized number of edges $\frac{|\mathcal{A}|}{n}$ for different algorithms and $N=20$.}
\label{edgesimfigcomp}
\end{figure}

In Fig. \ref{nodedegfig}, we show the probability mass functions (PMFs) corresponding to the degree of a given node of the network for different values of $N$. More specifically, let $I$ have a uniform PMF on the set $\{1,\ldots,n\}$. For a given $N$, the corresponding PMF at a given degree $d$ in Fig. \ref{nodedegfig} is then the probability that Node $I$ has degree $d$ in the final network topology provided by Algorithm 1.  We can also observe that almost all the nodes in the network have a degree of $6$ or less, which is in agreement with Theorem 2. Also, the fraction of nodes with degree $2$ increases as $N$ increases, and we expect it to approach to $1$ as $N\rightarrow\infty$ as a result of the aforementioned convergence to line topology.

\begin{figure}[h]
\center
\scalebox{0.55}{\includegraphics{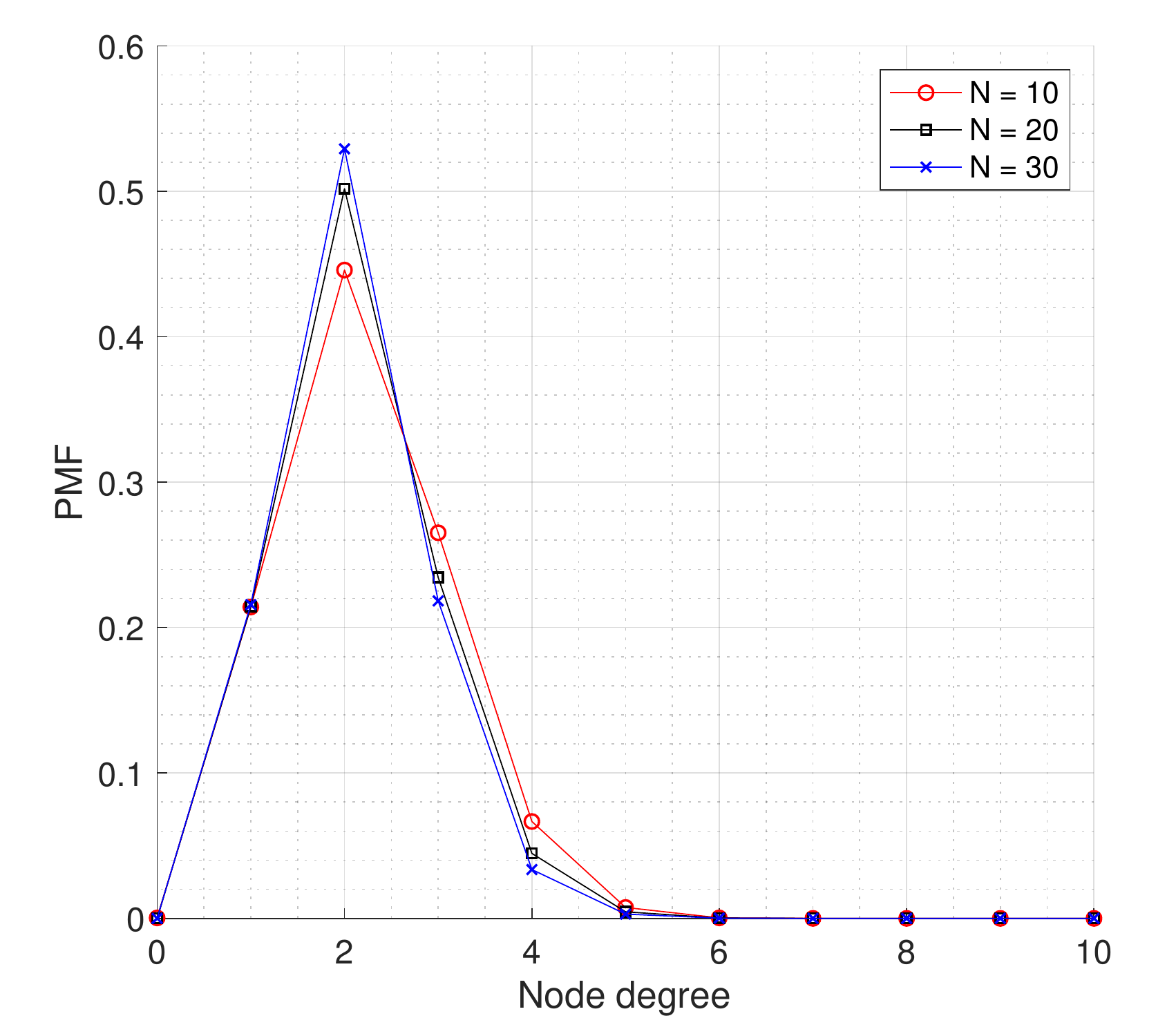}}
\caption{The PMFs of the degree of a given node for different node densities with Algorithm \ref{mainalgo}.}
\label{nodedegfig}
\end{figure}

We show the PMFs of the degrees of individual nodes and the corresponding expected node degrees in Fig. \ref{degsfordifferentindices} for the special case $N=20$. We can observe that nodes with lower or higher indices are more likely to have lower degrees. In fact, by design, the only way for a low-indexed node to ``gain'' degree in Algorithm 1 is by receiving connections from a higher-indexed node, and there can be at most $5$ such connections. Likewise, the only way for a high-indexed node to gain degree is by establishing connections to lower-indexed nodes. Similarly, there can be at most $5$ such connections. On the other hand, a medium-indexed node can potentially have many higher-indexed nodes as well as many lower-indexed neighboring nodes, implying a potential degree of $10$ in the final topology. Thus, intuitively, a medium-indexed node should typically have a larger average degree compared to a low- or high-indexed node. Fig. \ref{degsfordifferentindices} verifies this intuition.

\begin{figure}[h]
\center
\scalebox{0.55}{\includegraphics{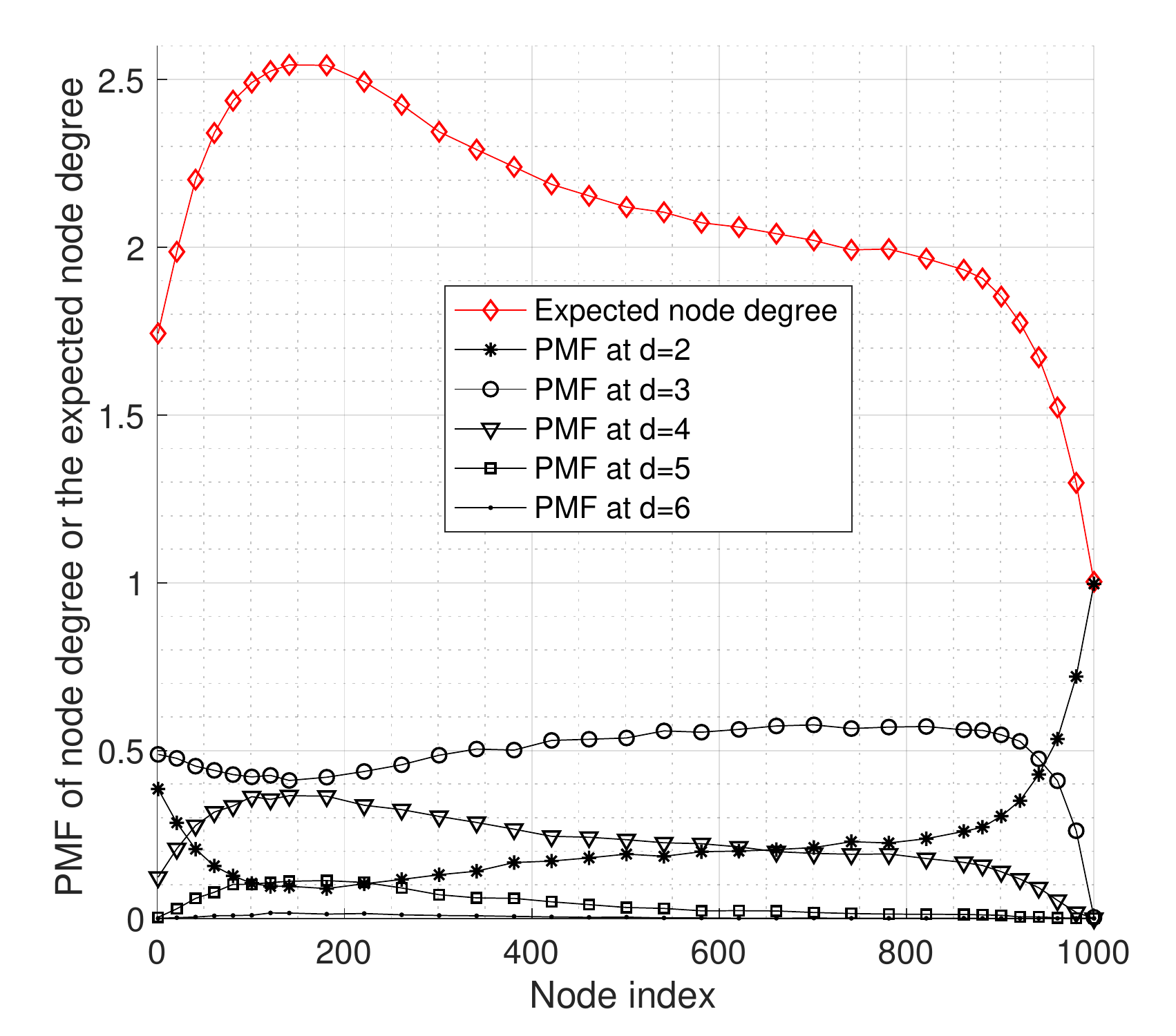}}
\caption{The expected value and the PMF of node degrees  for different node indices using Algorithm 1 and $N=20$.}
\label{degsfordifferentindices}
\end{figure}

In Fig. \ref{nodedegfigcomp}, we compare different algorithms with respect to their PMFs corresponding to the degree of a given node of the network. We consider a node density of $N=20$. Algorithm 1 provides the minimum possible average node degree of $2.12$, followed by an average degree of $2.50$ provided by XTC. Compared to the average of $18.63$ for the Gilbert graph, Algorithm 1 provides a $9$-fold reduction on the node degree.

\begin{figure}[h]
\center
\scalebox{0.55}{\includegraphics{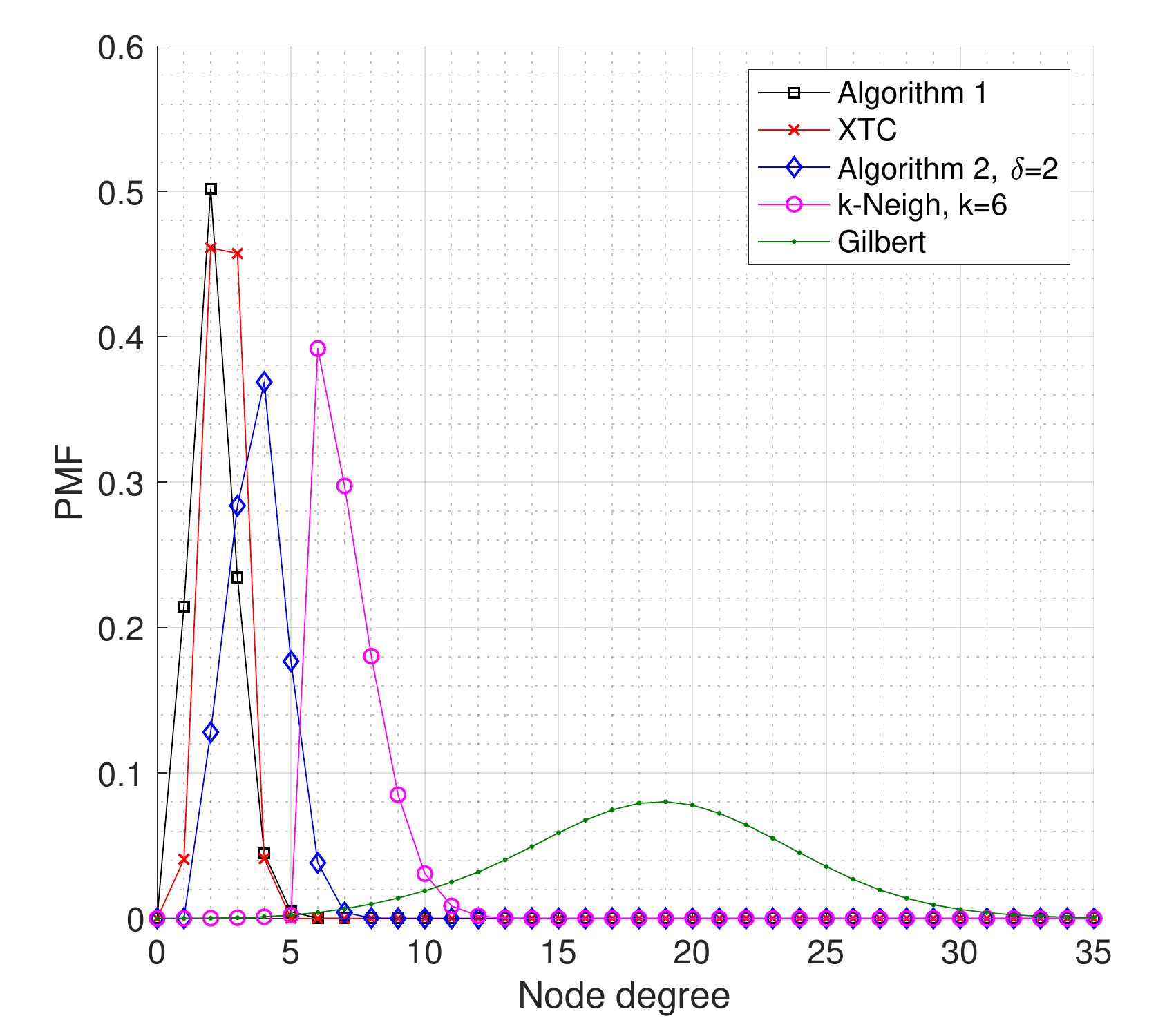}}
\caption{The PMFs of the degree of a given node for different algorithms and $N=20$.}
\label{nodedegfigcomp}
\end{figure}

Fig. \ref{maxnodedegfig} shows the PMFs of the maximum node degree of the network for Algorithm 1 and different node densities. According to Theorem 1, the PMFs should only take values on the set $\{0,\ldots,10\}$. In fact, for every value of $N$, we have not observed a single realization of node locations where the maximum node degree is $8$ or higher. Such realizations obviously exist (see Section II); Fig \ref{maxnodedegfig} rather suggests that they correspond to very rare events. Fig. \ref{maxnodedegfigcomp} provides the comparison of different algorithms in terms of the PMFs of the maximum node degrees. We can observe that the maximum degree with Algorithm 1 is more likely to be $5$ compared to a maximum degree of $4$ with XTC. The price to pay to guarantee a minimum degree of $2$ via Algorithm 2 is to increase the maximum node degree to $7$ with high probability. 

\begin{figure}[h]
\center
\scalebox{0.55}{\includegraphics{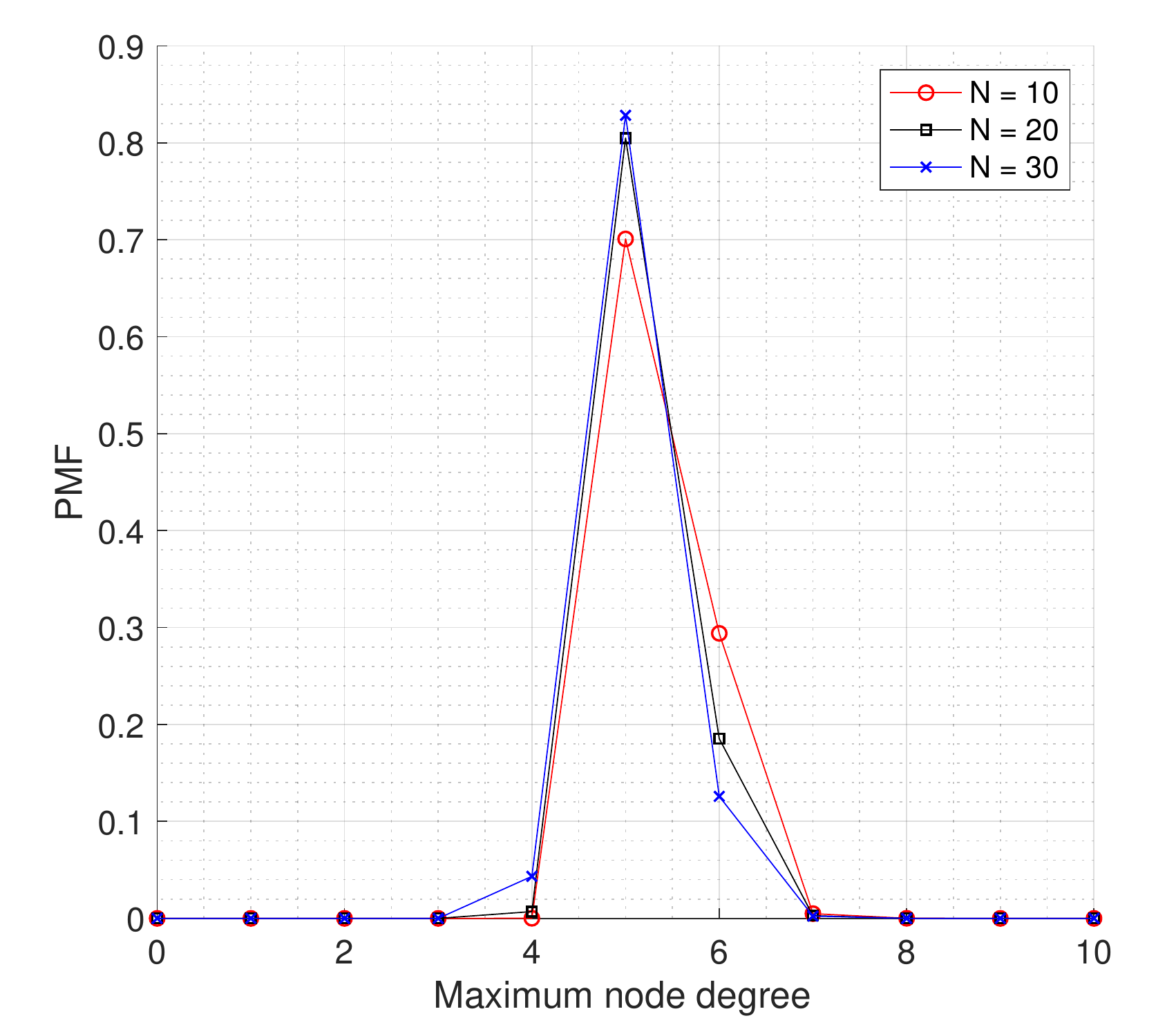}}
\caption{The PMFs of the maximum node degree for different node densities using Algorithm \ref{mainalgo}. }
\label{maxnodedegfig}
\end{figure}

\begin{figure}[h]
\center
\scalebox{0.55}{\includegraphics{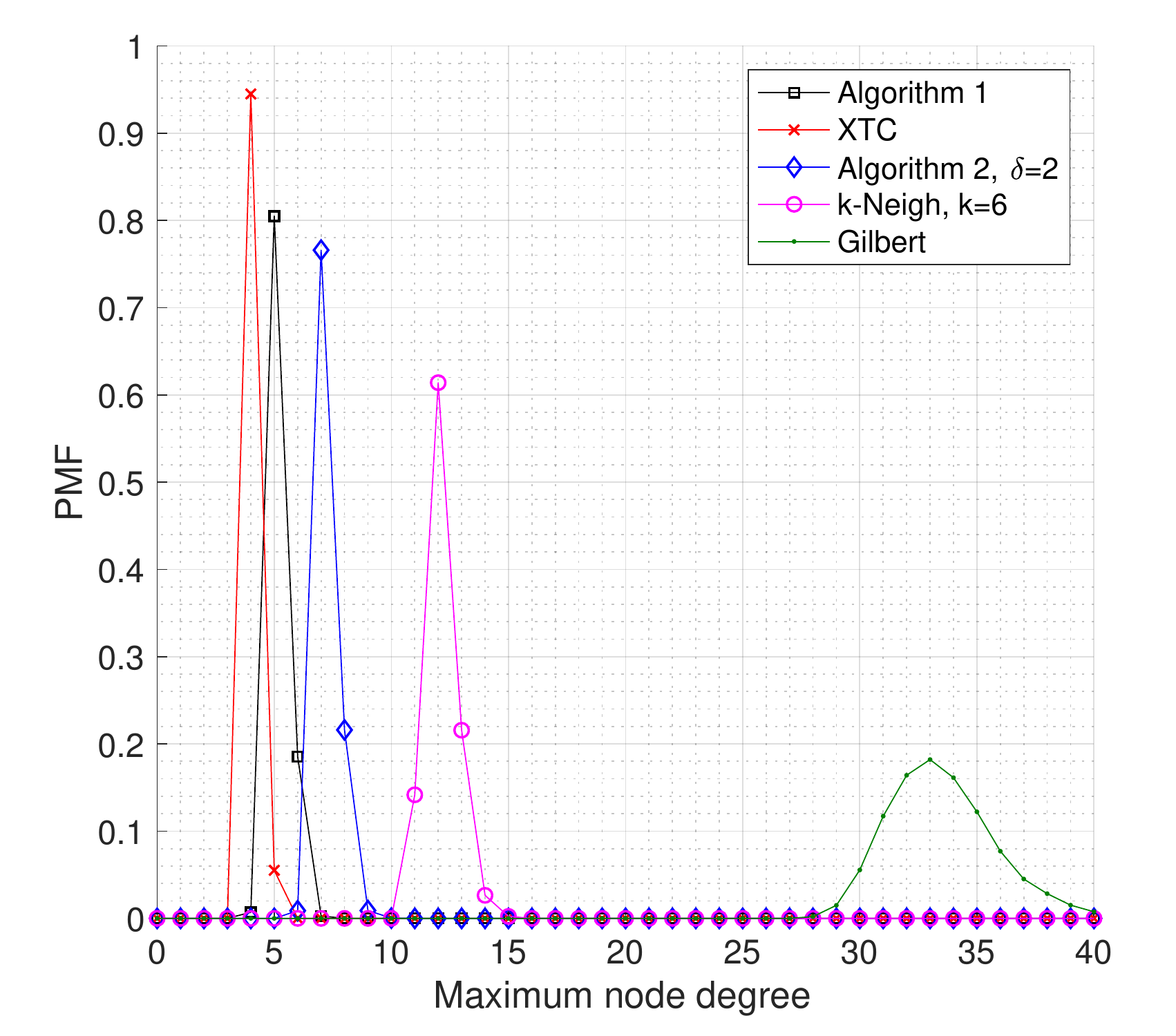}}
\caption{The PMFs of the maximum node degree for different algorithms and $N=20$.}
\label{maxnodedegfigcomp}
\end{figure}

In Fig. \ref{stretchsimfig}, we show the CDFs of the $\alpha$-stretch factors associated with two given distinct nodes of the network for $\alpha \in \{0,1,2\}$ (hop-, distance-, and power-stretch factors) and $N\in\{10,20,30\}$. Specifically, let $(I_1,I_2)$ have a uniform PMF on the set $\{(i,j):1\leq i < j \leq n\}$. Given $N$ and $\alpha$, the corresponding CDF evaluated at a given stretch factor $t$ in Fig. \ref{nodedegfig} is then given by $\mathrm{Pr}[c_{\alpha}(I_1,I_2;\mathcal{A}) \leq tc_{\alpha}(I_1,I_2;g(\mathcal{V}))]$. 

We can observe that all the CDFs remain less than $1$ at every finite stretch factor. This means that the algorithm cannot provide a constant $\alpha$-stretch factor for any $\alpha \geq 0$. This result is not surprising as by Theorems 3 and 4, any topology control algorithm that solely relies on neighborhood information will necessarily have unbounded stretch factors. Still, we can observe that our algorithm keeps the stretch factors low with high probability, at least for some cases of $\alpha$ and $N$. For example, for the network with $N=20$ that is connected for more than $99\%$ of the time, any two nodes that are $h$ hops away in the Gilbert graph will be no more than $5h$ hops away in $(\mathcal{V},\mathcal{A})$ for more than $90\%$ of the time. 

\begin{figure}[h]
\center
\scalebox{0.55}{\includegraphics{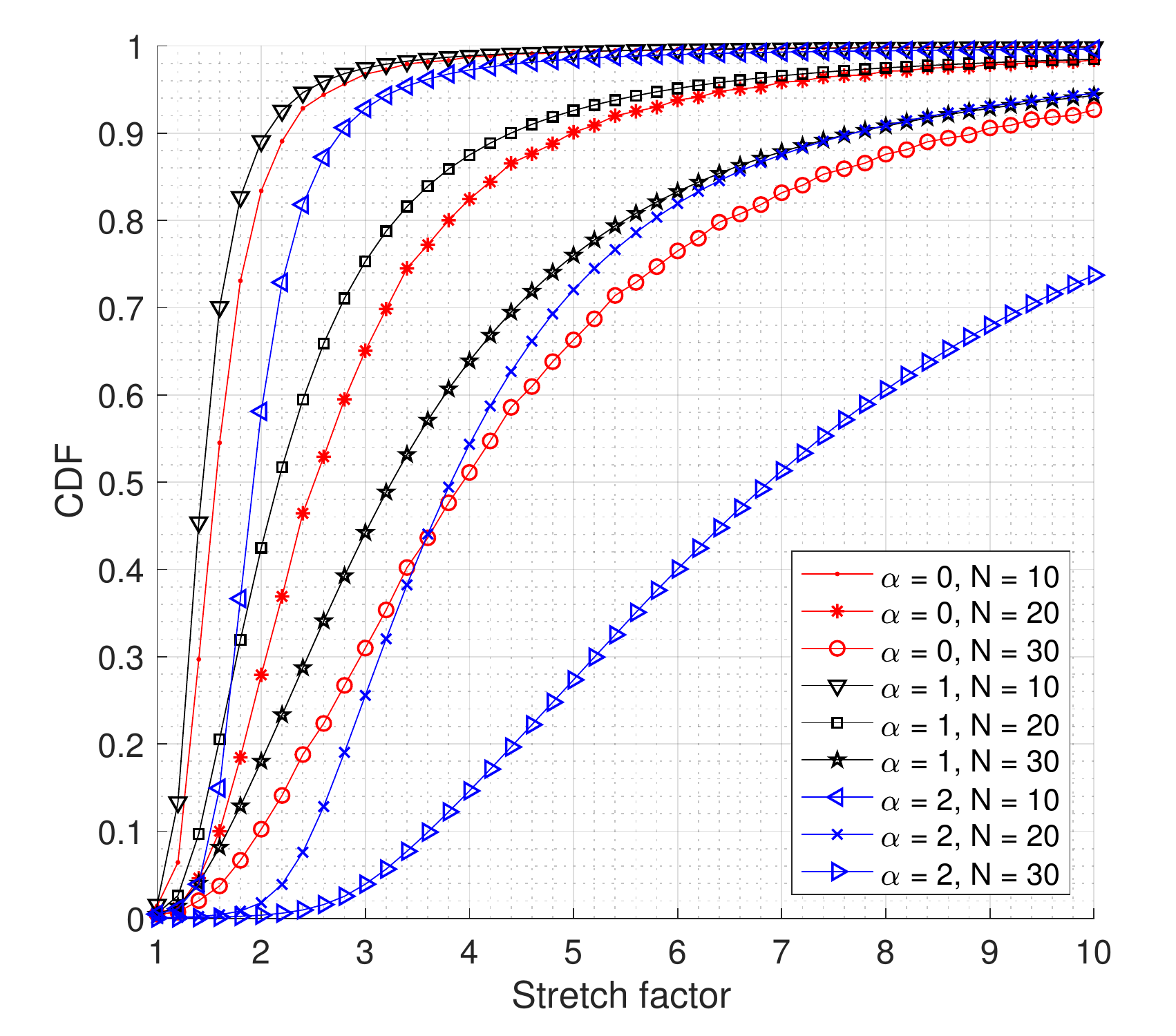}}
\caption{ The CDFs of $\alpha$-stretch factors for different node densities with Algorithm \ref{mainalgo}.}
\label{stretchsimfig}
\end{figure}

Comparison of different algorithms in terms of their stretch factors are provided in Fig. \ref{stretchsimfigcomp}. Typically, algorithms that result in more edges provide a better stretch factor distribution. For example, the $k$-Neigh algorithm outperforms all other algorithms for the case $\alpha\in\{1,2\}$. XTC also provides very good performance for $\alpha\in\{1,2\}$ despite providing a very sparse topology: It is only slightly worse than the $k$-Neigh topologies, outperforms the denser topologies provided by Algorithm 2 for $\delta = 2$, and significantly outperforms the sparser topology of Algorithm 1. Interestingly, for $\alpha = 0$, Algorithm 1, despite inducing a sparser topology compared to XTC, outperforms XTC by a significant margin for a wide range of stretch factors. Algorithm 2 for $\delta \!=\! 2$ outperforms even the much denser $k$-Neigh topology in certain cases. Therefore, Algorithms 1 and 2 can provide very good performance especially in terms of the $0$-stretch factors.

\begin{figure}[h]
\center
\scalebox{0.55}{\includegraphics{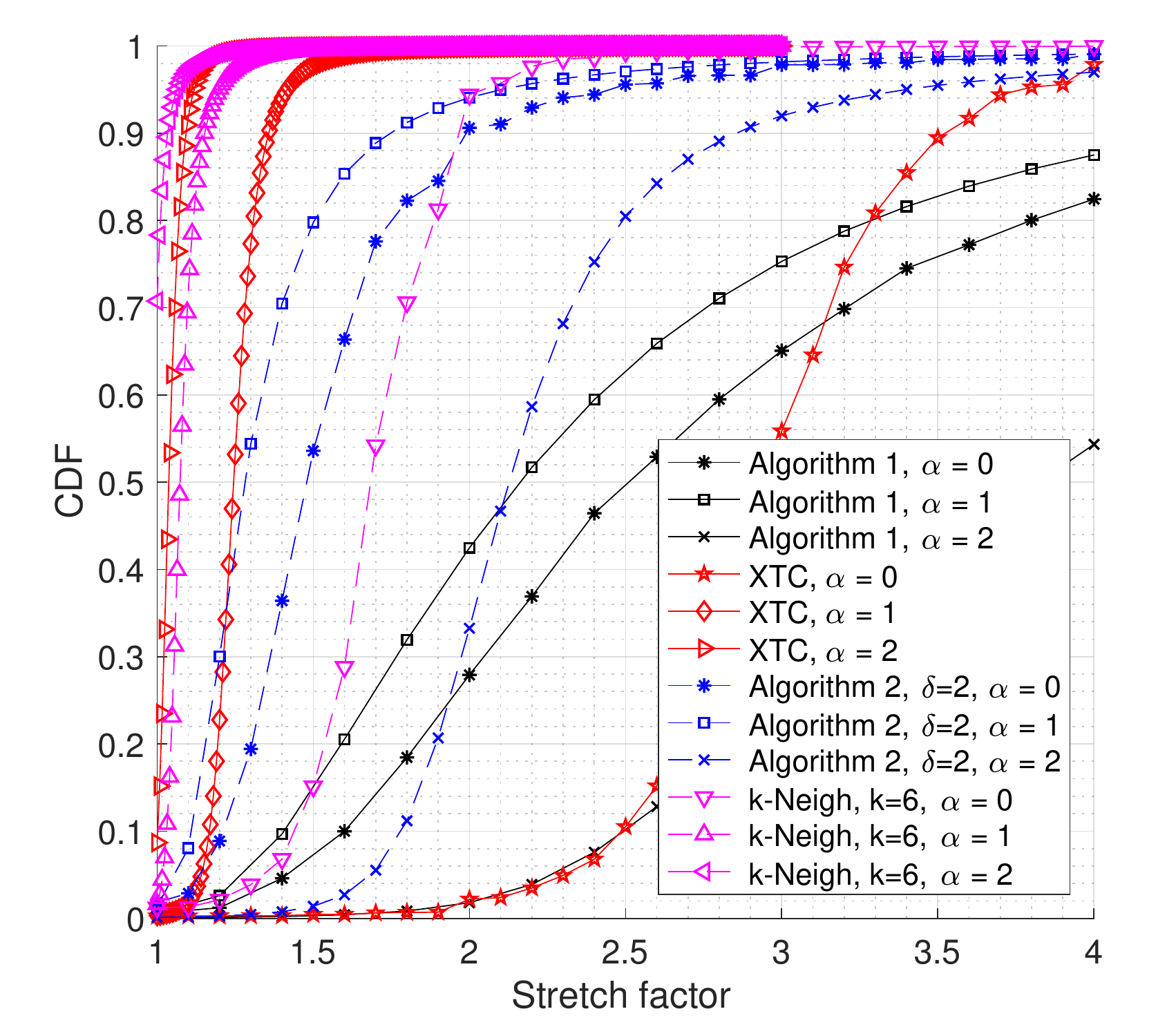}}
\caption{ The CDFs of $\alpha$-stretch factors for different algorithms and $N=20$.}
\label{stretchsimfigcomp}
\end{figure}

\section{Conclusions}
\label{secvi}
We have studied the problem of topology control in wireless ad-hoc networks consisting of $n$ nodes that are located on the plane. We have considered the disk-connectivity model where any two given neighboring nodes that lie within a certain communication range can be directly connected. We have addressed the fundamental problem of generating network topologies with the practically-relevant graph-theoretical properties such as connectivity or degree-boundedness.

We have observed that all the previous work in the literature require detailed geographical/locational information at each node to achieve these desired properties. We have shown that, in fact, a sufficient condition to achieve degree-bounded connectivity is just for each node to know the identification numbers of its one- and two-hop neighbors - no distance/directional information is needed whatsoever. Our corresponding local topology control algorithm guarantees a connected network with $5n$ edges and a maximum node degree of $10$. We have shown that for most networks, these numbers are in fact much lower. We have also designed an algorithm that can provide an upper bound as well as a lower bound on node degrees.

\begin{IEEEbiography}[{\includegraphics[width=1in,height=1.25in,clip,keepaspectratio]{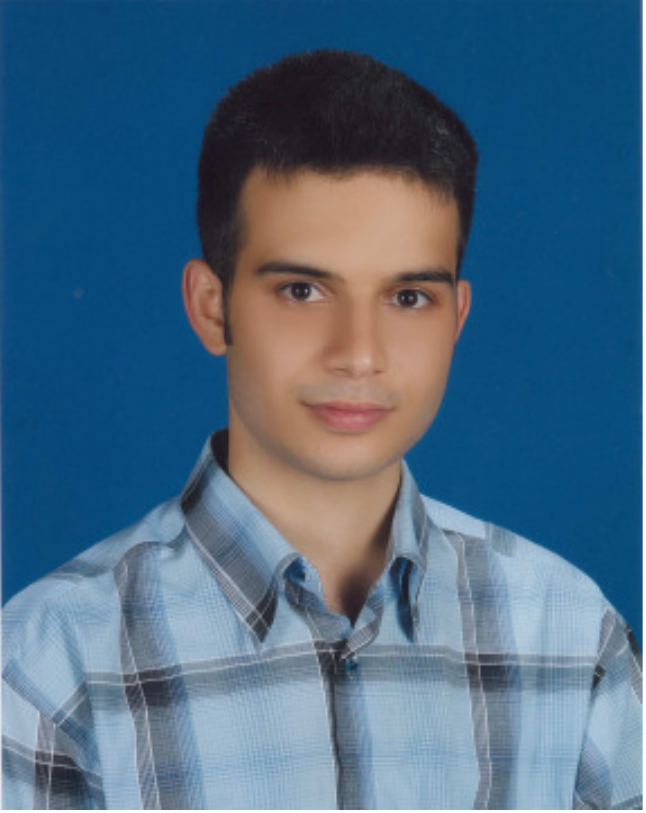}}]{Erdem Koyuncu} is an Assistant Professor at the Department of Electrical and Computer Engineering (ECE) of the University of Illinois at Chicago (UIC). He received the B.S. degree from the Department of Electrical and Electronics Engineering of Bilkent University in 2005. He received the M.S. and Ph.D. degrees in 2006 and 2010, respectively, both from the Department of Electrical Engineering and Computer Science of the University of California, Irvine (UCI). Between Jan. 2011 and Aug. 2016, he was a Postdoctoral Scholar at the Center for Pervasive Communications and Computing of UCI. Between Aug. 2016 and Aug. 2018, he was a Research Assistant Professor at the ECE Department of UIC. 
\end{IEEEbiography}

\begin{IEEEbiography}[{\includegraphics[width=1in,height=1.25in,clip,keepaspectratio]{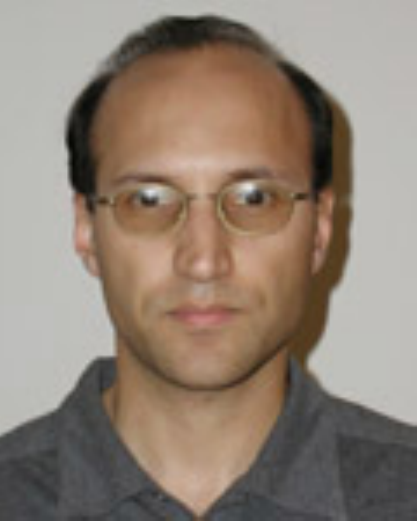}}]{Hamid Jafarkhani} is a Chancellor's Professor at the Department of Electrical Engineering and Computer Science,
University of California, Irvine, where he is also the Director of Center
for Pervasive Communications and Computing and the Conexant-Broadcom
Endowed Chair. He was a a Visiting Scholar at Harvard University in 2015 and a Visiting Professor at California Institute of Technology in 2018. Among his awards are the IEEE Marconi Prize Paper Award in Wireless Communications,  
the IEEE Communications Society Award for Advances in Communication, and the IEEE Eric E. Sumner Award. 

Dr. Jafarkhani is listed as a highly cited researcher in http://www.isihighlycited.com.
According to the Thomson Scientific, he is one of the top 10 most-cited researchers in the 
field of ``computer science'' during 1997-2007. 
He is the 2017 Innovation Hall of Fame Inductee at the University of Maryland's School of Engineering. 
He is a Fellow of AAAS, an IEEE Fellow, 
and the author of the book ``Space-Time Coding: Theory and Practice.''
\end{IEEEbiography}

\end{document}